\documentclass[a4paper,UKenglish,cleveref,autoref,thm-restate]{lipics-v2021}

\nolinenumbers

\crefname{section}{Sect.}{Sects.}
\Crefname{section}{Sect.}{Sects.}
\crefname{definition}{Def.}{Defs.}
\Crefname{definition}{Def.}{Defs.}
\crefname{appendix}{App.}{Appx.}
\Crefname{appendix}{App.}{Appx.}
\crefname{theorem}{Thm.}{Thms.}
\Crefname{theorem}{Thm.}{Thms.}
\crefname{lemma}{Lem.}{Lems.}
\Crefname{lemma}{Lem.}{Lems.}
\crefname{corollary}{Cor.}{Cors.}
\Crefname{corollary}{Cor.}{Cors.}
\crefname{algorithm}{Alg.}{Algs.}
\Crefname{algorithm}{Alg.}{Algs.}
\crefname{figure}{Fig.}{Figs.}
\Crefname{figure}{Fig.}{Figs.}
\crefname{example}{Ex.}{Exs.}
\Crefname{example}{Ex.}{Exs.}
\crefname{table}{Tbl.}{Tbls.}
\Crefname{table}{Tbl.}{Tbls.}

\usepackage{graphicx}
\usepackage{algorithmicx}
\usepackage{algorithm}
\usepackage[noend]{algpseudocode}
\usepackage{float}
\usepackage{booktabs}
\usepackage{amsmath,amssymb}
\usepackage{multirow}
\usepackage{nicematrix}
\usepackage{tablefootnote}
\usepackage{xspace}
\usepackage{microtype}
\usepackage{tikz}
\usepackage{mdframed}
\usepackage{nicefrac}
\usepackage{bm}
\usepackage{mathtools}
\usepackage{subcaption}
\usepackage{enumerate}

\usetikzlibrary{positioning, automata, calc}

\newcommand{\mvdv}[1]{}
\newcommand{\ms}[1]{}
\newcommand{\sj}[1]{}
\newcommand{\marck}[1]{}
\newcommand{\marnix}[1]{}
\newcommand{\sebastian}[1]{}
\newcommand{\sjminor}[1]{}

\newcommand{\ie}{\emph{i.e.}}
\newcommand{\eg}{\emph{e.g.}}

\newcommand{\unif}{\mathsf{unif}}
\newcommand{\dirac}{\mathsf{dirac}}
\newcommand{\assign}{\leftarrow}
\newcommand{\tuple}[1]{\langle #1 \rangle}
\newcommand{\mdp}{M}
\newcommand{\memdp}{\mathcal{M}}

\newcommand{\supp}{\mathit{Supp}}
\newcommand{\dist}[1]{\mathit{Dist}(#1)}
\newcommand{\last}{\mathit{last}}
\newcommand{\first}{\mathit{first}}

\newcommand{\powerset}{\mathcal{P}} %
\newcommand{\paths}{\mathsf{Path}}
\newcommand{\pathsfin}{\mathsf{Path}_\text{fin}}
\newcommand{\Path}{\pi}

\newcommand{\restrictenv}[2]{{#1}_{\downarrow#2}}

\newcommand{\reachable}[1]{\mathsf{Reachable}(#1)}
\newcommand{\reachablem}[2]{\mathsf{Reachable}_{#1}(#2)}

\newcommand{\initdist}{\iota}

\newcommand{\custombomdp}[1]{\mathcal{B}_{#1}}
\newcommand{\bomdp}{\custombomdp{\memdp}}

\newcommand{\genericlocal}{{\mathcal L}}

\newcommand{\pathbelief}{\mathsf{Belief}}
\newcommand{\stateset}{S}
\newcommand{\actionset}{A}
\newcommand{\transitionfunc}{p}
\newcommand{\transitionfuncs}{\{\transitionfunc_i\}_{i \in I}}

\newcommand{\infset}{\mathit{Inf}}
\newcommand{\mc}{C}
\newcommand{\memc}{\mathcal{C}}

\newcommand{\memdptuple}{\memdp = \tuple{\stateset, \actionset, \initdist, \transitionfuncs}}

\newcommand{\Finally}{\lozenge}
\newcommand{\Always}{\square}

\newcommand{\tTypeface}[1]{\mathsf{#1}}
\newcommand{\tP}{\tTypeface{P}}

\newcommand{\tPSPACE}{\tTypeface{PSPACE}}

\newcommand{\tEXPTIME}{\tTypeface{EXP}}

\newcommand{\tNL}{\tTypeface{NL}}

\newcommand{\tUD}{\tTypeface{UD}}

\newcommand{\tXP}{\tTypeface{XP}}

\DeclareMathAlphabet{\mathsf}{OT1}{\sfdefault}{m}{n}
\SetMathAlphabet{\mathsf}{bold}{OT1}{\sfdefault}{b}{n}
\newcommand{\bld}[1]{\boldsymbol{#1}}

\newcommand{\projbelief}[1]{\mathit{Env}(#1)}
\newcommand{\projstate}[1]{\mathit{St}(#1)}

\newcommand{\beliefupd}{\mathsf{BU}}
\newcommand{\allow}{\mathsf{Allow}}

\newcommand{\stratallow}{\strat_\mathsf{Allow}}

\newcommand{\stratwin}{\strat^*}
\newcommand{\strats}{\Sigma}

\newcommand{\stratbscc}{\strat_\mathsf{BSCC}}

\newcommand{\Bu}{B\"uchi\xspace}
\newcommand{\CoB}{Co-\Bu}

\newcommand{\Sbscc}{S_{\mathsf{B}}}
\newcommand{\Sbsccj}[1]{S_{\mathsf{B}#1}}
\newcommand{\Jbscc}{\mathcal{J}_{\tuple{s,J}}}
\newcommand{\Aallow}{A_{\mathsf{Allow}}}

\newcommand{\strat}{\sigma}
\newcommand{\fsc}{\mathcal{F}}

\newcommand{\creachobj}[1]{\Finally #1}
\newcommand{\csafeobj}[1]{\Always #1}
\newcommand{\cbuchiobj}[1]{\Always \Finally #1}
\newcommand{\ccobuchiobj}[1]{\Finally \Always #1}

\newcommand{\reachobj}{\creachobj{T}}
\newcommand{\safeobj}{\csafeobj{T}}
\newcommand{\buchiobj}{\cbuchiobj{T}}
\newcommand{\cobuchiobj}{\ccobuchiobj{T}}

\algnewcommand{\IfThenElse}[3]{%
  \State \algorithmicif\ #1\ \algorithmicthen\ #2\ \algorithmicelse\ #3}
\algdef{SE}[DOWHILE]{Do}{doWhile}{\algorithmicdo}[1]{\algorithmicwhile\ #1}%

\newcommand{\winregion}[2]{\mathit{Win}_{#1}({#2})}
\newcommand{\winset}{S_{\mathit{win}}}

\newcommand{\rabinObj}{\Phi}
\newcommand{\locRabinObj}{\rabinObj^{\mathit{Loc}}}
\newcommand{\clocRabinObj}[1]{\rabinObj^{\mathit{CLoc}}(#1)}
\newcommand{\rabinPair}[1]{\tuple{\rabincobuchi_{#1}, \rabinbuchi_{#1}}}
\newcommand{\rabinObjDef}{\rabinObj = \{ \rabinPair{i} \mid 1 \le i \le k. \,\, \rabinbuchi_i \subseteq \rabincobuchi_i \subseteq S \}}

\newcommand{\rabinbuchi}{\mathfrak{C}}
\newcommand{\rabincobuchi}{\mathfrak{B}}

\newcommand{\liftedstrat}{\hat{\strat}}
\newcommand{\liftedobj}{\hat{\rabinObj}}
\newcommand{\newinit}[2]{#1^#2}
\newcommand{\winningfrontier}{\textit{WF}}
\newcommand{\localbomdp}[2]{#1\{#2\}}

\newcommand{\staterestrict}[2]{{#1}_{\csafeobj{#2}}}

\newcommand{\stateremove}[2]{\mathit{StateRemove}(#1, #2)}
\newcommand{\rabinWinSet}{W_{\locRabinObj}}

\newcommand{\algreach}{\textsc{Reach}}
\newcommand{\rabinalgname}{\textsc{Rabin}}
\newcommand{\defeq}{:=}
\newcommand{\genericalgoname}{\textsc{Check}}

\newenvironment{proofsketch}{%
  \renewcommand{\proofname}{Proof sketch}\proof}{\endproof}

\title{A PSPACE Algorithm for Almost-Sure Rabin~Objectives in~Multi-Environment~MDPs} 
\author{Marnix Suilen}{Radboud University, Nijmegen, the Netherlands}{marnix.suilen@ru.nl}{https://orcid.org/0000-0003-2163-3504}{}
\author{Marck van der Vegt}{Radboud University, Nijmegen, the Netherlands}{marck.vandervegt@ru.nl}{https://orcid.org/0000-0003-2451-5466}{}
\author{Sebastian Junges}{Radboud University, Nijmegen, the Netherlands}{sebastian.junges@ru.nl}{https://orcid.org/0000-0003-0978-8466}{}
\authorrunning{M.~Suilen, M.~van der Vegt, S.~Junges}
\Copyright{Marnix Suilen, Marck van der Vegt and Sebastian Junges}

\begin{document}

\maketitle

\begin{abstract}
Markov Decision Processes (MDPs) model systems with uncertain transition dynamics. 
Multiple-environment MDPs (MEMDPs) extend MDPs. 
They intuitively reflect finite sets of MDPs that share the same state and action spaces but differ in the transition dynamics. 
The key objective in MEMDPs is to find a single policy that satisfies a given objective in every associated MDP. 
The main result of this paper is PSPACE-completeness for almost-sure Rabin objectives in MEMDPs.
This result clarifies the complexity landscape for MEMDPs and contrasts with results for the more general class of partially observable MDPs (POMDPs), where almost-sure reachability is already EXPTIME-complete, and almost-sure Rabin objectives are undecidable.
\end{abstract}

\section{Introduction}\label{sec:introduction}
Markov decision processes (MDPs) are the ubiquitous model for decision-making under uncertainty~\cite{DBLP:books/wi/Puterman94}.  
An elementary question in MDPs concerns the existence of strategies that satisfy qualitative temporal properties, such as \emph{is there a strategy such that the probability of reaching a set of target states is one}?
Qualitative properties in MDPs have long been considered as pre-processing for probabilistic model checking of quantitative properties~\cite{Baier2008,DBLP:conf/sfm/ForejtKNP11}. 
Recently, however, qualitative properties have received interest in the context of shielding~\cite{DBLP:conf/aaai/AlshiekhBEKNT18,DBLP:conf/cav/JungesJS20}, \ie, the application of model-based reasoning to ensure safety in reinforcement learning~\cite{DBLP:journals/jmlr/GarciaF15,DBLP:books/lib/SuttonB98}. 

An often prohibitive assumption in using MDPs is that the strategy can depend on the precise state. 
To follow such a strategy, one must precisely observe the state of the system, \ie, of an agent and its environment. 
The more general partially observable MDPs (POMDPs)~\cite{DBLP:journals/ai/KaelblingLC98} do not make this assumption. 
In POMDPs, a strategy cannot depend on the precise states of the system but only on the (sequence of) observed labels of visited states. 
As a consequence, and in contrast to MDPs, winning strategies may require memory. 
Indeed, the existence of strategies that satisfy qualitative objectives on MDPs is efficiently decidable in polynomial time using standard graph-algorithms~\cite{DBLP:conf/soda/ChatterjeeJH04,DBLP:conf/mfcs/ChatterjeeDH10}. 
In contrast, in POMDPs, deciding almost-sure reachability is already $\tEXPTIME$-complete~\cite{DBLP:conf/fossacs/BaierBG08,DBLP:conf/mfcs/ChatterjeeDH10}, and the existence of strategies for a more general class of almost-sure Rabin objectives is undecidable~\cite{DBLP:conf/fossacs/BaierBG08}. 

Multi-environment MDPs (MEMDPs)~\cite{DBLP:conf/fsttcs/RaskinS14} %
model a finite set of MDPs, called \emph{environments}, that share the state space but whose transition relation may be arbitrarily different. 
For any given objective, the key decision problem asks to find a single strategy that satisfies the objective in every associated MDP. 
MEMDPs are particularly suitable to model settings where one searches for a winning strategy that is robust to perturbations or random initialization problems. 
Examples of MEMDPs range from code-breaking games such as Mastermind, card games such as Free-cell, or Minesweeper, to more serious applications in robotics~\cite{DBLP:conf/icra/ChenFHL16}, \eg, high-level planning where an artifact with unknown location must be recovered. 

MEMDPs are special POMDPs~\cite{DBLP:conf/aips/ChatterjeeCK0R20}: An agent can observe the current state but not the transition relation determining the outcomes of its actions. 
This type of partial observation has an important effect: When an agent observes the next state, it may rule out that a certain MDP describes the true system.
Thus, the set of MDPs that may describe the system is monotonically decreasing~\cite{DBLP:conf/aips/ChatterjeeCK0R20}. 
We call this property as \emph{monotonic information gain}.

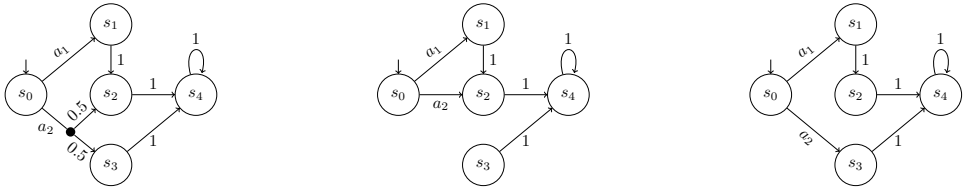
\begin{figure}[t]
    \centering
    \begin{subfigure}[t]{0.3\textwidth}
    \resizebox{0.7\textwidth}{!}{%
        \begin{tikzpicture}
    \node[state] (s0) {$s_0$};
    \node[above=0.3cm of s0] (init){};
    \node[state, right=0.9cm of s0] (s2) {$s_2$};
    \node[state, above=0.6cm of s2] (s1) {$s_1$};
    \node[state, below=0.6cm of s2] (s3) {$s_3$};
    \node[state, right=0.9cm of s2] (s4) {$s_4$};

    \node[circle, inner sep=2pt, fill=black, below right =0.4cm and 0.55cm of s0] (a2) {};

    \draw[-] (s0) --node[below] {$a_2\quad$} (a2);

    \draw[->] (s0) --node[sloped,above]{$a_1$} (s1);
    \draw[->] (a2) --node[sloped,above]{$0.5$} (s2);
    \draw[->] (a2) --node[sloped,below]{$0.5$} (s3);
    \draw[->] (s1) --node[right]{$1$} (s2);
    \draw[->] (s2) --node[above]{$1$} (s4);
    \draw[->] (s3) --node[below]{$1$} (s4);

    \draw[->] (s4) edge[loop above]node[above]{$1$} (s4);

    \draw[->] (init) -- (s0);
\end{tikzpicture}
    }%
    \end{subfigure}%
    \hfill
    \begin{subfigure}[t]{0.3\textwidth}
        \resizebox{0.7\textwidth}{!}{%
        \begin{tikzpicture}
    \node[state] (s0) {$s_0$};
    \node[above=0.3cm of s0] (init){};
    \node[state, right=0.9cm of s0] (s2) {$s_2$};
    \node[state, above=0.6cm of s2] (s1) {$s_1$};
    \node[state, below=0.6cm of s2] (s3) {$s_3$};
    \node[state, right=0.9cm of s2] (s4) {$s_4$};

    \draw[->] (s0) --node[sloped,above]{$a_1$} (s1);
    \draw[->] (s0) --node[sloped, below]{$a_2$}(s2);
    \draw[->] (s1) --node[right]{$1$} (s2);
    \draw[->] (s2) --node[above]{$1$} (s4);
    \draw[->] (s3) --node[below]{$1$} (s4);

    \draw[->] (s4) edge[loop above]node[above]{$1$} (s4);

    \draw[->] (init) -- (s0);
\end{tikzpicture}
            }%
    \end{subfigure}%
    \hfill
    \begin{subfigure}[t]{0.3\textwidth}
        \resizebox{0.7\textwidth}{!}{%
        \begin{tikzpicture}
    \node[state] (s0) {$s_0$};
    \node[above=0.3cm of s0] (init){};
    \node[state, right=0.9cm of s0] (s2) {$s_2$};
    \node[state, above=0.6cm of s2] (s1) {$s_1$};
    \node[state, below=0.6cm of s2] (s3) {$s_3$};
    \node[state, right=0.9cm of s2] (s4) {$s_4$};

    \draw[->] (s0) --node[sloped,above]{$a_1$} (s1);
    \draw[->] (s0) --node[sloped, below]{$a_2$}(s3);
    \draw[->] (s1) --node[right]{$1$} (s2);
    \draw[->] (s2) --node[above]{$1$} (s4);
    \draw[->] (s3) --node[below]{$1$} (s4);

    \draw[->] (s4) edge[loop above]node[above]{$1$} (s4);

    \draw[->] (init) -- (s0);
\end{tikzpicture}
            }%
    \end{subfigure}
    \caption{A MEMDP with three environments.}
    \label{fig:memdp:example}
\end{figure}

We illustrate MEMDPs in \Cref{fig:memdp:example}. 
This MEMDP consists of three environments.
When playing action $a_1$ in state $s_0$, we end up in state $s_1$ in every environment.
If we play action $a_2$ and observe that we end up in state $s_2$, we can infer that we have to either be in the left or the middle environment, while observing $s_3$ after $a_2$ rules out the middle environment.
Such information cannot be lost in a MEMDP, hence \emph{monotonic} information gain.

\begin{table}[b]
\centering
\setlength\tabcolsep{0.25em}
 \caption{Known complexity (completeness) for MEMDPs, new results are in boldface. 
  $\tNL$ and $\tEXPTIME$ denote the classes $\mathsf{NLOGSPACE}$ and $\mathsf{EXPTIME}$, and $\tUD$ denotes $\mathsf{UNDECIDABLE}$. 
 }
 \label{tbl:results}
\resizebox{\textwidth}{!}{
 \begin{NiceTabular}{l|lllll|lll}
 \toprule
 Semantics & \multicolumn{5}{c}{Almost-sure} & \multicolumn{3}{c}{Possible}\\
 Model & ~MDP~ & $2$-MEMDP & $k$-MEMDP & MEMDP & POMDP & ~MDP~  & MEMDP & POMDP \\ \midrule
Reachability & $\tP$~\cite{DBLP:conf/mfcs/ChatterjeeDH10} & $\tP$~\cite{DBLP:conf/fsttcs/RaskinS14} & $\bld{\tP}$~\Cref{cor:2memdps} & $\tPSPACE$~\cite{DBLP:conf/tacas/VegtJJ23}                          & $\tEXPTIME$~\cite{DBLP:conf/fossacs/BaierBG08,DBLP:conf/mfcs/ChatterjeeDH10} & $\tNL$~\cite{DBLP:conf/mfcs/ChatterjeeDH10} & $\bld{\tNL}$~\Cref{thm:pos:equaltomdps} & ${\tNL}$~\cite{DBLP:conf/mfcs/ChatterjeeDH10}\\ 
Safety       & $\tP$~\cite{DBLP:conf/mfcs/ChatterjeeDH10} & $\tP$~\cite{DBLP:conf/fsttcs/RaskinS14} & $\bld{\tP}$~\Cref{cor:2memdps}   & $\bld{\tPSPACE}$~\Cref{thm:rabin_pspace}                   & $\tEXPTIME$~\cite{DBLP:conf/fsttcs/BerwangerD08,DBLP:journals/lmcs/RaskinCDH07} & $\tP$~\cite{DBLP:conf/mfcs/ChatterjeeDH10}  & $\bld{\tP}$~\Cref{thm:pos:equaltomdps}  & $\tEXPTIME$~\cite{DBLP:conf/mfcs/ChatterjeeDH10}\\  
Büchi        & $\tP$~\cite{DBLP:conf/mfcs/ChatterjeeDH10} & $\tP$~\cite{DBLP:conf/fsttcs/RaskinS14} & $\bld{\tP}$~\Cref{cor:2memdps}   & $\bld{\tPSPACE}$~\Cref{thm:rabin_pspace}                     & $\tEXPTIME$~\cite{DBLP:conf/fossacs/BaierBG08,DBLP:conf/mfcs/ChatterjeeDH10} & $\tP$~\cite{DBLP:conf/mfcs/ChatterjeeDH10}  & $\bld{\tP}$~\Cref{thm:pos:equaltomdps}  & $\tUD$~\cite{DBLP:conf/fossacs/BaierBG08}\\ 
Co-Büchi     & $\tP$~\cite{DBLP:conf/mfcs/ChatterjeeDH10} & $\tP$~\cite{DBLP:conf/fsttcs/RaskinS14} & $\bld{\tP}$~\Cref{cor:2memdps}   & $\bld{\tPSPACE}$~\Cref{thm:rabin_pspace}                   & $\tUD$~\cite{DBLP:conf/fossacs/BaierBG08} & $\tP$~\cite{DBLP:conf/mfcs/ChatterjeeDH10}  & $\bld{\tP}$~\Cref{thm:pos:equaltomdps}  & $\tEXPTIME$~\cite{DBLP:conf/mfcs/ChatterjeeDH10}\\ 
Parity       & $\tP$~\cite{DBLP:conf/mfcs/ChatterjeeDH10} & $\tP$~\cite{DBLP:conf/fsttcs/RaskinS14} & $\bld{\tP}$~\Cref{cor:2memdps}   & $\bld{\tPSPACE}$~\Cref{thm:rabin_pspace} & $\tUD$~\cite{DBLP:conf/fossacs/BaierBG08} & $\tP$~\cite{DBLP:conf/mfcs/ChatterjeeDH10}  & $\bld{\tP}$~\Cref{thm:pos:equaltomdps}  & $\tUD$~\cite{DBLP:conf/fossacs/BaierBG08}\\ 
Rabin & $\tP$~\cite{DBLP:conf/icalp/ChatterjeeAH05} & $\bld{\tP}$~\Cref{cor:2memdps} & $\bld{\tP}$~\Cref{cor:2memdps}  & $\bld{\tPSPACE}$~\Cref{thm:rabin_pspace} & $\tUD$~\cite{DBLP:conf/fossacs/BaierBG08} & $\tP$~\cite{Baier2008} & $\bld{\tP}$~\Cref{thm:pos:equaltomdps} & $\tUD$~\cite{DBLP:conf/fossacs/BaierBG08}\\
\bottomrule
 \end{NiceTabular}
 }
\end{table}

The most relevant results for qualitative properties on MEMDPs are by Raskin and Sankur~\cite{DBLP:conf/fsttcs/RaskinS14}, and by Van~der~Vegt et al.~\cite{DBLP:conf/tacas/VegtJJ23}. 
The former paper focuses on the case with only \emph{two environments}, which we refer to as $2$-MEMDPs (and more generally, $k$-MEMDPs). 
It shows, among others, that almost-sure parity can be decided in polynomial time. 
In $2$-MEMDPs, the memory for a winning strategy is polynomial in the size of the MEMDP, while for arbitrary environments, winning strategies for almost-sure reachability may be exponential~\cite{DBLP:conf/tacas/VegtJJ23}. 
However, despite the need for exponential strategies, almost-sure \emph{reachability} in MEMDPs is decidable in $\tPSPACE$ via a recursive algorithm that exploits the aforementioned monotonic information gain~\cite{DBLP:conf/tacas/VegtJJ23}.

The main result of this paper is a landscape of qualitative Rabin objectives and their subclasses in MEMDPs, see~\Cref{tbl:results}. %
The key novelty is a $\tPSPACE$ algorithm to decide the existence of strategies in MEMDPs that satisfy an almost-sure Rabin objective. 
The algorithm relies on two key ingredients:
First, as shown in \Cref{sec:belief:sufficient}, for almost-sure Rabin objectives, a particular type of finite-memory strategies (with memory exponential in the number of environments) is sufficient, in contrast to POMDPs. 
Second, towards an algorithm, we observe that a traditional, per Rabin-pair, approach for Rabin objectives does not generalize to MEMDPs (\cref{sec:nonlocalrabin}). 
It does, however, generalize to what we call belief-local MEMDPs, in which one, intuitively, cannot gain any information (\cref{sec:local_rabin_alg}).  
Exploiting the monotonic information gain of (general) MEMDPs, we construct a recursive algorithm with polynomial stack size, inspired by~\cite{DBLP:conf/tacas/VegtJJ23}, that solves Rabin objectives in belief-local MEMDPs (\cref{sec:recursive_computation}).
Finally, we establish $\tPSPACE$-hardness for almost-sure safety and clarify that for possible objectives, MEMDPs can be solved as efficiently as MDPs. 
Proofs are in the appendix.

\subsection*{Related Work}
\label{sec:relatedwork}

Besides almost-sure objectives, Raskin and Sankur~\cite{DBLP:conf/fsttcs/RaskinS14} also study \emph{limit-sure} objectives for MEMDPs of two environments. 
Where almost-sure objectives require that the satisfaction probability of the objective equals one, an objective is satisfied \emph{limit-surely} whenever for any $\epsilon > 0$, there is a strategy under which the objective is satisfied with probability at least $1-\epsilon$.
For $2$-MEMDPs, limit-sure parity objectives are decidable in P~\cite{DBLP:conf/fsttcs/RaskinS14}.

Closely related to the study of almost-sure objectives in MEMDPs and POMDPs is the \emph{value $1$ problem} for probabilistic automata (PA).
A PA can be seen as a POMDP where all states have the same observation and are thus indistinguishable. 
The value $1$ problem is to decide whether the supremum of the acceptance probability over all words equals one.
This problem is undecidable for general PA~\cite{DBLP:conf/icalp/GimbertO10}, but recent works have studied several subclasses of PA for which the value $1$ problem is decidable.
Most notably, \emph{$\#$-acyclic} PA~\cite{DBLP:conf/icalp/GimbertO10}, \emph{structurally simple} and \emph{simple} PA~\cite{DBLP:conf/lics/ChatterjeeT12}, and \emph{leaktight} PA~\cite{DBLP:conf/lics/FijalkowGO12}.
Leaktight PA are the most general of these subclasses~\cite{DBLP:journals/corr/Fijalkow14}. 
They contain the others, and the value $1$ problem is $\tPSPACE$-complete~\cite{DBLP:journals/corr/FijalkowGKO15}.

The interpretation of MEMDPs as a special POMDP is successfully used in the quantitative setting, where the goal is to find a strategy that maximizes the probability of reaching a target. Finding a strategy that maximizes the finite-horizon expected reward in MEMDPs is  $\tPSPACE$-complete~\cite{DBLP:conf/aaai/ChadesCMNSB12}, as is also the case for the same problem in more general POMDPs~\cite{DBLP:journals/mor/PapadimitriouT87}.

Besides a special class of POMDP, MEMDPs are also a class of \emph{robust} MDP with discrete uncertainty sets~\cite{DBLP:journals/ior/NilimG05,DBLP:journals/mor/Iyengar05}.
In the robotics and AI communities MEMDPs are studied in that context, primarily for quantitative objectives such as maximizing discounted reward or regret minimization~\cite{DBLP:conf/aaai/RigterLH21}.
\emph{Parametric} MDPs (pMDPs) are another formalism for defining MDPs with a range of transition functions~\cite{DBLP:conf/birthday/0001JK22}.
Where we seek a single strategy that is winning for all environments, parameter synthesis is often about finding a single parameter instantiation (or: environment) such that all strategies are winning~\cite{DBLP:journals/tac/CubuktepeJJKT22,DBLP:conf/cav/AndriushchenkoC21}.
That is, the quantifiers are reversed. 
A notable exception is work on quantitative properties in pMDPs by Arming et al.~\cite{DBLP:conf/qest/ArmingBCKS18}, which interprets a parametric MDP as a MEMDP and solves it as a POMDP.
With the (altered) quantifier order in pMDPs, memoryless deterministic strategies are sufficient: 
The complexity of finding a parameter instantiation such that under all (memoryless deterministic) strategies a quantitative reachability objective is satisfied is in NP when the number of parameters is fixed and ETR-complete in the general case~\cite{DBLP:conf/concur/WinklerJPK19}. 
Determining whether a memoryless deterministic policy is robust is both ETR and co-ETR-hard~\cite{DBLP:conf/concur/WinklerJPK19}.

\section{Background and Notation}\label{sec:background}
Let $\mathbb{N}$ denote the natural numbers.
For a set $X$, the powerset of $X$ is denoted by $\powerset(X)$, and the disjoint union of two sets $X,Y$ is denoted $X \sqcup Y$.
A discrete probability distribution over a finite set $X$ is a function $\mu \colon X \to [0,1]$ with $\sum_{x \in X} \mu(x) = 1$, the set of all discrete probability distributions over $X$ is $\dist{X}$.
The support of a distribution $\mu \in \dist{X}$ is the set of elements $x$ with $\mu(x) >0$ and is denoted by $\supp(\mu)$.
We denote the uniform distribution over $X$ by $\unif(X)$ and the Dirac distribution with probability $1$ on $x$ by $\dirac(x)$.

\subsection{Markov Decision Processes}\label{subsec:mdps}
We briefly define standard (discrete-time) Markov decision processes and Markov chains.

\begin{definition}[MDP]\label{def:MDP}
  A \emph{Markov decision process} (MDP) is a tuple $\mdp \defeq \tuple{\stateset, \actionset, \initdist, \transitionfunc}$ where $\stateset$ is a finite set of states and $\initdist \in \stateset$ is the initial state, $\actionset$ is the finite set of \emph{actions}, and $\transitionfunc \colon \stateset \times \actionset \rightharpoonup \dist{\stateset}$ is the partial probabilistic \emph{transition function}.
  By $\actionset(s)$, we denote the set of enabled actions for $s$, which are the actions for which $\transitionfunc(s, a)$ is defined.
\end{definition}
For readability, we write $\transitionfunc(s,a,s')$ for  $\transitionfunc(s,a)(s')$.
A \emph{path} in an MDP is a sequence of successive states and actions, $\Path = s_0 a_0 s_1 a_1 \ldots \in (SA)^*S$, such that $s_0 = \initdist$, $a_i \in \actionset(s)$, and $p(s_i,a_i,s_{i+1}) > 0$ for all $i \geq 0$, and we write $\overline{\Path}$ for only the sequence of states in $\Path$.
The probability of following a path $\Path$ in an MDP with transition function $\transitionfunc$ is defined as $\transitionfunc(\Path) = \transitionfunc(s_0a_0\dots) = \prod_{i=0} \transitionfunc(s_i, a_i, s_{i+1})$.
The set of all (finite) paths on an MDP $\mdp$ is $\paths(\mdp)$ (resp.\ $\pathsfin(\mdp)$).
Whenever clear from the context, we omit the MDP $\mdp$ from these notations.
We write $\first(\Path)$ and $\last(\Path)$ for the first and last state in a finite path, respectively, and the concatenation of two paths $\pi_1, \pi_2$ is written as $\pi_1 \cdot \pi_2$.
The set of \emph{reachable states} from $S' \subseteq S$ is $\reachable{S'} \defeq \{ s' \in S \mid \exists \Path \in \pathsfin\colon \first(\Path) \in S', \last(\Path) = s' \}$.
A state $s\in S$ is a \emph{sink state} if $\reachable{\{s\}}=\{s\}$.
An MDP is acyclic if each state is a sink state or not reachable from its successor states.
The \emph{underlying graph} of an MDP is a tuple $\tuple{V,E}$ with vertices $V \defeq \{v_s \mid s \in S\}$ are edges $E \defeq \{ \tuple{v_s, v_{s'}} \mid \forall s,s' \in S\colon \exists a \in A\colon \transitionfunc(s,a,s') > 0 \}$.

A \emph{sub-MDP} of an MDP $\mdp = \tuple{S,A,\initdist,\transitionfunc}$ is a tuple $\tuple{S',A', \initdist',\transitionfunc'}$
with states $\emptyset \neq S' \subseteq S$, actions $\emptyset \neq A' \subseteq A$, initial state $\initdist' \in S'$, and a transition function $\transitionfunc'$ such that $\forall s \in S\colon \emptyset \neq A'(s) \subseteq A(s)$ and $\forall s,s' \in S', a \in A'(s)\colon \supp(\transitionfunc(s,a)) \subseteq S'$ and $\transitionfunc'(s,a,s') \defeq \transitionfunc(s,a,s')$.  
An \emph{end-component} of an MDP $\mdp$ is a sub-MDP where $\reachable{S'} = S'$. 
Sub-MDPs and end-components are standard notions; for details, cf.~\cite{DBLP:phd/us/Alfaro97,Baier2008,DBLP:conf/fsttcs/RaskinS14}.

A \emph{Markov chain} is an MDP where there is only one action available at every state: $\forall s \in S. \, |A(s)| = 1$.
We write an MC as a tuple $\mc \defeq \tuple{\stateset, \initdist, \transitionfunc}$ where $\stateset$ is a set of states, $\initdist \in \stateset$ is the initial state, and $\transitionfunc \colon \stateset \to \dist{\stateset}$ is the transition function.
Paths in MCs are sequences of successive states, and their underlying graph is analogously defined as for MDPs.
A subset $T \subseteq S$ is strongly connected if for each pair of states $\tuple{s,s'} \in T$ there exists a finite path $\Path$ with $\first(\Path) = s$ and $\last(\Path) = s'$.
A \emph{strongly connected component} (SCC) is a strongly connected set of states $T$ such that no proper superset of $T$ is strongly connected.
A \emph{bottom SCC} (BSCC) is an SCC $S'$ where no state $s \in S \setminus S'$ is reachable.

\subsection{Strategies and Objectives}\label{subsec:properties}
We now formally define strategies and their objectives.
Strategies resolve the action choices in MDPs.
A strategy is a measurable function $\strat \colon \pathsfin \to \dist{\actionset}$ such that for all finite paths $\Path \in \pathsfin$ we have $\supp(\strat(\last(\Path))) \subseteq A(\last(\Path))$.
A strategy is deterministic if it maps only to Dirac distributions, and it is memoryless if the action (distribution) only depends on the last state of every path.
We write $\strats$ for the set of all strategies.

A strategy $\strat$ applied to an MDP $\mdp$ induces an infinite-state MC $\mdp[\strat] = \tuple{S^*, \initdist, \transitionfunc_{\strat}}$ such that for any path $\Path$: $\transitionfunc_{\strat}(\overline{\Path}, \overline{\Path} \cdot s') = \transitionfunc(\last(\Path), \strat(\Path), s')$. 
This MC has a probability space with a unique probability measure $\mathbb{P}_{M[\strat]}$ via the cylinder construction~\cite{Baier2008,DBLP:journals/corr/abs-2305-10546}. 

A strategy is a \emph{finite-memory} strategy if it can be encoded by a stochastic Moore machine, also known as a finite-state controller (FSC)~\cite{DBLP:conf/uai/MeuleauKKC99}.
An FSC is a tuple $\fsc = \tuple{N, n_{\iota}, \alpha, \eta}$, where $N$ is a finite set of \emph{memory nodes}, the \emph{initial node} $n_\iota \in N$, $\alpha \colon \stateset \times N \to \dist{\actionset}$ is the \emph{action mapping}, and $\eta \colon N \times \actionset \times \stateset \to \dist{N}$ is the \emph{memory update function}.
The induced MC $\mdp[\fsc]$ of an MDP $\mdp$ and finite-memory strategy represented by an FSC $\fsc$ is finite and defined by the following product construction:
$\mdp[\fsc] = \tuple{S \times N, \tuple{\initdist, n_{\iota}}, \transitionfunc_{\fsc}}$, where 
$
\transitionfunc_{\fsc}(\tuple{s,n},\tuple{s',n'}) = \sum_{a \in A} \alpha(s,n,a) \cdot \transitionfunc(s,a,s') \cdot \eta(n,a,s',n')$.
A strategy is \emph{memoryless} if its FSC representation has a single memory node, \ie, $|N| =1$.

We consider both \emph{almost-sure} and \emph{possible} objectives for MDPs and MCs with state space $S$.
An \emph{objective} $\Phi$ is a measurable subset of $P \subseteq S^\omega$.
An MC $\mc$ is almost-surely (or possibly) \emph{winning} for an objective $\Phi$ iff $\mathbb{P}_{\mc}(\paths(\mc) \cap \Phi) = 1$ (or $\mathbb{P}_{\mc}(\paths(\mc) \cap \Phi) > 0$).
A state $s$ is winning whenever the MC with its initial state replaced by $s$ is winning.
We write $\mc \models \Phi$ and $s \models^{\mc} \Phi$ to denote that MC $\mc$ and state $s$ are winning for $\Phi$. 

\begin{definition}[Winning]\label{def:winning}
    An MDP $\mdp$ is winning for $\Phi$ if there exists a strategy $\strat \in \strats$ such that the induced MC $\mdp[\strat]$ is winning for $\Phi$, and the strategy $\strat$ is then also called winning.
\end{definition}
Like above, we denote winning in MDPs with $M \models \Phi$ or $s \models^M \Phi$, respectively. Sometimes, we explicitly add the winning strategy and write $\mdp[\strat] \models \Phi$ and $s \models^{\mdp[\strat]} \Phi$ for the MDP winning $\Phi$ under $\strat$ from its initial state or some other state $s$, respectively. 

\begin{definition}[Winning Region]\label{def:winning region}
We call the set of states of an MDP (or MC) that are winning objective $\Phi$ the \emph{winning region}, denoted as $\winregion{\mdp}{\Phi} = \{ s \in S \mid s \models^{\mdp} \Phi \}$.
\end{definition}
We define Rabin objectives.
Let $\mc = \tuple{S, \initdist, \transitionfunc}$ be a MC with associated probability measure $\mathbb{P}_{\mc}$, $\Path \in \paths(\mc)$ a path, and $\infset(\Path) \subseteq S$ the set of states reached infinitely often along $\Path$.

\begin{definition}[Rabin objective]
    A Rabin objective is a set of Rabin pairs: $\rabinObjDef$.
        A path $\Path \in \paths(\mc)$ wins $\rabinObj$ if there is a Rabin pair $\rabinPair{i}$ in $\rabinObj$ where the path leaves the states in $\rabincobuchi_i$ only finitely many times, and states in $\rabinbuchi_i$ are visited infinitely often.
        The MC $\mc$ wins a Rabin objective almost-surely (or possibly) if 
        $
        \mathbb{P}_{\mc}(\Path \in \paths(\mc) \mid \exists \rabinPair{i} \in \Phi\colon \infset(\Path) \subseteq \rabincobuchi_i \wedge \infset(\Path) \cap \rabinbuchi_i \neq \emptyset) = 1$ (or possibly when $> 0$).
\end{definition}
For MDPs, almost-sure and possible Rabin objectives can be solved in polynomial time~\cite[Thm.\ 10.127]{Baier2008}, and the strategies are memoryless deterministic~\cite[Thm.\ 4]{DBLP:conf/icalp/ChatterjeeAH05}.
Other objectives, specifically reachability ($\reachobj$), safety ($\safeobj$), \Bu ($\buchiobj$), \CoB ($\cobuchiobj$), and parity are included in Rabin objectives~\cite{DBLP:conf/icalp/ChatterjeeAH05} for a set $T \subseteq S$.
\Cref{appx:aux:sec:background} contains formal definitions.

\section{Multi-Environment MDPs and the Problem Statement}\label{subsec:memdps}
Next, we introduce the multi-environment versions of MDPs and MCs.
Intuitively, these can be seen as finite sets of MDPs and MCs that share the same states and actions.

\begin{definition}[MEMDP]\label{def:MEMDP}
  A \emph{multi-environment MDP} (MEMDP) is a tuple $\memdptuple$ with $\stateset, \actionset, \initdist$ as for MDPs, and $\transitionfuncs$ is a \emph{finite set of transition functions}, where $I$ are the \emph{environment} indices. 
  We also write $\memdp = \{\mdp_i\}_{i \in I}$ as a set of MDPs,
  where $M_i = \tuple{S, A, \initdist, p_i}$.
\end{definition}
For a MEMDP $\memdp$ and a set $I' \subseteq I$, we define the restriction to $I'$ as the MEMDP $\restrictenv{\memdp}{I'} = \tuple{\stateset, \actionset, \initdist, \{p_{i}\}_{i \in I'}}$.
To change the initial state of $\memdp$, we define $\memdp^{\initdist'} = \tuple{\stateset, \actionset, \initdist', \transitionfuncs}$.

A \emph{multi-environment MC} (MEMC) is a MEMDP with $\forall s \in S\colon |A(s)| = 1$.
A MEMC is a tuple $\memc = \tuple{\stateset, \initdist, \transitionfuncs}$ or equivalently a set of MCs $\memc = \{\mc_i\}_{i \in I}$.
A BSCC in a MEMC %
is a set of states $S' \subseteq S$ such that $S'$ forms a BSCC in every MC $\mc_i \in \memc$.
The underlying graph of a MEMDP or MEMC is the disjoint union of the graphs of the environments.

Similarly to \Cref{def:winning} for MDPs, a strategy $\strat$ for a MEMDP $\memdp$ is winning for objective $\Phi$ if and only if the induced MEMC $\memdp[\strat] = \{\memdp[\strat]_i \}_{i \in I}$ is winning in all environments: $\forall i \in I\colon \memdp[\strat]_i$ is winning for $\Phi$.
Winning regions, \Cref{def:winning region}, extend similarly to MEMDPs: $\winregion{\memdp}{\Phi} = \{s \in S \mid s \models^\memdp \Phi\}$.

The central decision problem in this paper is:
\begin{mdframed}[linewidth=2pt,topline=false,rightline=false,bottomline=false]
Given a MEMDP $\memdp$ and a Rabin objective $\Phi$, is there a winning strategy for $\Phi$ in $\memdp$.
\end{mdframed}
We assume MEMDPs are encoded as an explicit list of MDPs and each MDP is given by the explicit transition function.
The value of the probabilities are not relevant.

We first consider possible semantics in MEMDPs, completing~\cref{tbl:results}.
For POMDPs, \CoB objectives are known to be undecidable~\cite{DBLP:conf/fossacs/BaierBG08}. 
We show that for various objectives, deciding them in MEMDPs is equally hard as in their MDP counterparts.

\begin{restatable}[]{theorem}{thmposequaltomdps}
\label{thm:pos:equaltomdps}
   Deciding possible reachability objectives for MEMDPs is in $\tNL$.
   Deciding possible safety, \Bu, \CoB, parity and Rabin objectives for MEMDPs is in $\tP$.
\end{restatable}\noindent
Using results on MDPs from~\cite{DBLP:conf/mfcs/ChatterjeeDH10}, these upper bounds are tight.
The main observation for membership is that a MEMDP is winning possibly objectives iff each MDP is possibly winning, due to a randomization over the individual winning strategies. We can then construct algorithms that solve each environment sequentially to answer the query on MEMDPs.

From here on, we focus exclusively on the almost-sure objectives.
\begin{theorem}
Almost-sure reach, safety, (co-)B\"uchi, and Rabin objectives for MEMDPs are $\tPSPACE$-hard.
\end{theorem}
This theorem follows from~\cite{DBLP:conf/tacas/VegtJJ23}, which shows that
almost-sure reachability is $\tPSPACE$-complete.
$\tPSPACE$-hardness of almost-sure safety can be established by minor modifications to the proof:
In particular, the $\tPSPACE$-hardness proof for reachability operates on acyclic MEMDPs, where we may reverse the target and non-target states to change the objective from almost-sure reachability to safety.
$\tPSPACE$-hardness of almost-sure (co-)B\"uchi, parity, and Rabin objectives follows via reduction from almost-sure reachability.

\section{Belief-Based Strategies are Sufficient}\label{sec:belief:sufficient}

In this section, we fix a MEMDP $\memdp = \tuple{S, A, \initdist, \transitionfuncs}$ with an almost-sure Rabin objective $\Phi$, and constructively show a more refined version of the following statement.
\begin{corollary}
    For a MEMDP $\memdp$ and an almost-sure Rabin objective $\Phi$, if there exists a winning strategy for $\Phi$, there also exists a winning finite-memory strategy $\stratwin$ such that the finite-state controller (FSC) for $\strat^{*}$ is exponential (only) in the number of environments.
\end{corollary}
This corollary to \cref{thm:beliefsufficient} below immediately gives rise to EXPTIME algorithms for the decision problem that simply iterate over all strategies. 
The particular shape of the strategy, a notion that we call \emph{belief-based}, will be essential later to establish $\tPSPACE$ algorithms.

\subsection{Beliefs in MEMDPs}
It is helpful to consider the strategy as a model for a decision-making agent. 
Then, in MEMDPs, the agent observes in which state $s \in S$ it currently is but does not have access to the environment $i \in I$. 
The hiding of environments gives rise to the notion of a \emph{belief-distribution}
in MEMDPs, akin to beliefs in partially observable MDPs~\cite{DBLP:journals/ai/KaelblingLC98}.
A belief distribution in a MEMDP is a probability distribution over environments $\mu \in \dist{I}$ that assigns a probability to how likely the agent is operating in each environment.
As we show below, we only need to consider the \emph{belief-support}, \ie, a subset of environments that keeps track of whether it is \emph{possible} that the agent operates in those environments. 
From now on, we shall simply write \emph{belief} instead of \emph{belief-support}.

\begin{definition}[Belief, belief update]
Given a finite path $\Path$, we define its \emph{belief} as its last state together with the set of environments for which this path has positive probability:
$\pathbelief(\Path)= \tuple{ \last(\Path), \{ i \in I \mid p_i(\Path) > 0 \}}$.
For path $\Path \cdot a s'$, the belief can be characterized recursively by the \emph{belief update} function $\beliefupd \colon S  \times \powerset(I) \times A \times S \to S \times \powerset(I)$. 
Let $\tuple{s,J} = \pathbelief(\Path)$, then:
$
\tuple{s', J'} = \beliefupd(\tuple{s,J}, a, s') \defeq \pathbelief(\Path \cdot a s'),
$
where $J' = \{j \in J \mid \transitionfunc_j(s,a,s') > 0\}$.
We also liberally write $\beliefupd(\tuple{s,J},a)$ for the set of beliefs $\tuple{s',J'}$ that are possible from $\tuple{s,J}$ via action $a$, and define the two projection functions: $\projstate{\tuple{s, J}} = s$ and $\projbelief{\tuple{s, J}} = J$.    
\end{definition}
Key to MEMDPs is the notion of \emph{revealing transitions}~\cite{DBLP:conf/fsttcs/RaskinS14}.
A revealing transition is a tuple $\tuple{s,a,s'}$ such that there exist two environments $i,i' \in I, i \neq i'$ with $\transitionfunc_i(s,a,s') > 0$ and $\transitionfunc_{i'}(s,a,s') = 0$.
Intuitively, a transition is revealing whenever observing this transition reduces the belief over environments the agent is currently in, since we observed a transition that is not possible in one or more environments.
From this notion of revealing transitions immediately follows the property of monotonic information gain in MEMDPs.

\begin{corollary}\label{lemma:belief-stabalizing}
Let $\Path \cdot as'$ be a finite path. 
Then: (1)~$\projbelief{\pathbelief(\Path \cdot as')} \subseteq \projbelief{\pathbelief(\Path)}$, and (2)~If there are environments $j,j' \in \projbelief{\pathbelief(\Path)}$  with $\transitionfunc_j(\last(\Path),a,s') > 0= 
    \transitionfunc_{j'}(\last(\Path),a,s'))$,
    then 
 $j' \not\in \projbelief{\pathbelief(\Path \cdot as')}$ and thus $\projbelief{\pathbelief(\Path \cdot as')} \subset \projbelief{\pathbelief(\Path)}$.
\end{corollary}

\label{sec:belief_based_strategies}
\noindent Key to our analysis of MEMDPs is the notion of \emph{belief-based} strategies.

\begin{definition}[Belief-based strategy]
    A strategy $\strat$ is \emph{belief-based} when for all finite paths $\Path$, $\Path'$ such that $\pathbelief(\Path)=\pathbelief(\Path')$ and $\last(\Path)=\last(\Path')$ implies  $\strat(\Path)=\strat(\Path')$.
Belief-based strategies can also be written as a function $\strat \colon S \times \powerset(I) \to \dist{A}$. 
\end{definition}
Belief-based strategies are a form of finite-memory strategies and are representable by FSCs.

\begin{restatable}[]{lemma}{lembeliefstratfsc}\label{lem:beliefstrat:fsc}
    A belief-based strategy $\strat \colon S \times \powerset(I) \to \dist{A}$ for a MEMDP $\memdp$ with states $S$ and actions $A$ can be represented by an FSC.
\end{restatable}\noindent
Similar to how states may be winning, beliefs can also be winning.
\begin{definition}[Winning belief]
    We call the belief $\tuple{s,J}$ \emph{winning for objective $\Phi$ in $\memdp$}, written as $\tuple{s, J} \models^{\memdp} \Phi$, if there exists a strategy $\strat \colon \paths \to \dist{A}$ such that for every environment $j \in J$, the induced MC is winning.
    That is,
    $
    \forall j \in J\colon \memdp_j[\strat] \models \Phi
    $.
\end{definition}%
    A MEMDP $\memdp$ is winning, $\memdp \models \Phi$, iff the initial belief $\tuple{\initdist, I}$ is winning.
We can extend the notion of a winning region to beliefs.
The (belief) winning region of a MEMDP $\memdp$ is $\winregion{\memdp}{\Phi} = \{ \tuple{s, J} \in S \times \powerset(I) \mid \tuple{s, J} \models^{\memdp} \Phi \}$.
The notion of winning beliefs has been used before in the context of POMDPs with other almost-sure objectives~\cite{DBLP:journals/ai/ChatterjeeCGK16}.

The induced MEMC of a MEMDP and FSC conservatively extends the standard product construction between an MDP and FSC to be applied to each transition function $\transitionfuncs$ individually. 
In that MEMC, the objective must be lifted to the new state space $\stateset \times N$.

\begin{definition}[Lifted Rabin objective]\label{def:lifted:rabin:objective}
    For Rabin objective $\rabinObjDef$ on $S$, the \emph{lifted Rabin objective} to  $S \times N$ is $\liftedobj \defeq \{ \tuple{\rabinbuchi_i \times N, \rabincobuchi_i \times N} \mid 1 \le i \le k \}$.
\end{definition}
When clear from the context, we implicitly apply this lifting where needed.

\subsection{Constructing Winning Belief-Based Strategies}

We are now ready to state the main theorem of this section.

\begin{restatable}[]{theorem}{thmbeliefsufficient}\label{thm:beliefsufficient}
For  MEMDP $\memdp$ and Rabin objective $\Phi$, there exists a winning strategy $\strat$ for $\Phi$ iff there exists a belief-based strategy $\stratwin$ that is winning for $\Phi$. 
\end{restatable}\noindent
The remainder of this section is dedicated to the necessary ingredients to prove \Cref{thm:beliefsufficient}.

\begin{definition}[Allowed actions]\label{def:allowed:actions}
The set of \emph{allowed actions} for a winning belief $\tuple{s, J}$ is
$
\allow(\tuple{s,J}) \defeq \{ a \in \actionset(s) \mid \forall \tuple{s', J'} \in \beliefupd(\tuple{s, J}, a)\colon \tuple{s', J'} \models^{\memdp} \Phi \}$.
\end{definition}
That is, an action at a winning belief is allowed if all possible resulting successor beliefs are still winning.
Using allowed actions we define the belief-based strategy $\stratallow \colon S \times \powerset(I) \to \dist{A}$:
\[\textstyle
\stratallow(\tuple{s, J}) \defeq \begin{cases}
    \unif(\allow(\tuple{s,J})) & \text{ if } \allow(\tuple{s,J}) \neq \emptyset, \\
    \unif(A(s)) & \text{ otherwise.}
\end{cases}
\]
The strategy $\stratallow$ randomizes uniformly over all allowed actions when the successor beliefs are still winning and over all actions when the belief cannot be winning.
We now sketch how to use $\stratallow$ to construct a winning belief-based strategy. See \Cref{appx:constructing:belief:based:strategies} for the details.

When playing $\stratallow$, the induced MEMC $\memdp[\stratallow]$ will almost-surely end up in a BSCC. 
Given belief $\tuple{s,J}$, we compute all environments $j \in J$ for which $\memdp_j[\stratallow]$ is in a BSCC~$\Sbscc$:
$
\Jbscc \defeq \{ j \in J \mid \exists \Sbscc \subseteq S \times \powerset(I)\colon \Sbscc \text{ is a BSCC in MC } \memdp[\stratallow]_j \wedge \tuple{s,J} \in \Sbscc \}$.

As a consequence, since $\stratallow$ is a belief-based strategy, every BSCC of the MEMC $\memdp[\stratallow]$ has a fixed set of environments that cannot change anymore.
As $\stratallow$ uniformly randomizes over all allowed actions, and it remains possible to win, there has to exist a strategy $\stratbscc \colon \Sbscc \to \dist{A}$ that is a sub-strategy of $\stratallow$, \ie, it only plays a subset of actions that are also played by $\stratallow$.
We construct appropriate sub-MDPs to compute $\stratbscc$ for a belief $\tuple{s,J}$ that is almost-surely winning for $\Phi$.
Using $\stratallow$ and $\stratbscc$, we construct the following belief-based strategy and show it is indeed winning. 
The other direction follows since beliefs are based on paths, which proves \cref{thm:beliefsufficient}. 
\[\textstyle
\stratwin(\tuple{s,J}) \defeq \begin{cases}
    \stratallow(\tuple{s,J}) & \text{ if } \Jbscc = \emptyset,\\
    \stratbscc(\tuple{s,J}) & \text{ otherwise}.  
\end{cases}
\]
\Cref{thm:beliefsufficient} shows that belief-based strategies are sufficient for almost-sure Rabin objectives in MEMDPs, which is not true on more general POMDPs~\cite{DBLP:journals/jcss/ChatterjeeCT16}.
Key is the monotonic information gain in MEMDPs, a property that POMDPs do not have in general~\cite{DBLP:conf/aips/ChatterjeeCK0R20}.

\begin{remark*}
    For the remainder of this paper, we assume all strategies are finite-memory strategies, and the induced (ME)MCs are defined via the product construction from  \Cref{subsec:properties}.
\end{remark*}

\section{Explicitly Adding Belief to MEMDPs}\label{sec:bomdps} %
Above, we showed that it is sufficient for a strategy to reason over the beliefs.
Now, we show how to add beliefs to the MEMDP, yielding a \emph{belief observation MDP} (BOMDP). We discuss their construction (\cref{subsec:bomdps}) and then algorithms on BOMDPs for reachability (\cref{subsec:reachabilitybomdps}) and so-called Safe \Bu objectives (\cref{sec:safe_buchi}). 
We discuss Rabin objectives in \cref{sec:pspace:algorithm}. 

\subsection{Belief-Observation MDPs}
\label{subsec:bomdps}
We create a product construction between the MEMDP and the beliefs $\powerset(I)$ such that the beliefs of the MEMDP $\memdp$ are directly encoded in the state space:
\begin{definition}[BOMDP]\label{def:bomdp}
  The \emph{belief observation MDP} (BOMDP) of MEMDP $\memdp = \tuple{\stateset, \actionset, \initdist, \transitionfuncs}$ 
  is a MEMDP $\bomdp = \tuple{\stateset', \actionset, \initdist', \{p'_i\}_{i \in i}}$
  with states $\stateset' = S \times \powerset(I)$,
  initial state $\initdist' = \tuple{\initdist, I}$, and
  partial transition functions that are defined when $a \in A(s)$ such that
  \[\textstyle
  p'_j(\tuple{s, J}, a, \tuple{s', J'}) = 
  \begin{cases}
      p_{j}(s, a, s') & \text{ if } \tuple{s', J'} = \beliefupd(\tuple{s, J},{a},{s'}) \wedge j \in J,\\
      0 & \text{ otherwise}.
  \end{cases} %
  \]
\end{definition}
BOMDPs are special MEMDPs; hence, all definitions for MEMDPs apply to BOMDPs.
Due to the product construction, a belief-support-based strategy for $\memdp$ can be turned into a memoryless strategy for $\bomdp$, and vice versa.
In BOMDPs, the belief $J$ is already part of the state, so we simplify the satisfaction notation to %
$\tuple{s, J} \models^{\bomdp} \Phi$ instead of $\tuple{\tuple{s, J}, J} \models^{\bomdp} \Phi$.

\begin{definition}[Lifted strategy]
   Given MEMDP $\memdp$ and a belief-based strategy $\strat$. %
   The \emph{lifted memoryless strategy} $\liftedstrat\colon (S \times \powerset(I)) \to \dist{A}$ on $\bomdp$ is  $\liftedstrat(s, J) \defeq \strat(\tuple{s, J})$.
\end{definition}
This lifting ensures that belief-based strategies and their liftings to BOMDPs coincide.

\begin{restatable}[]{lemma}{lemstratcorrespondence}\label{lem:strat_correspondence}
    Given a MEMDP $\memdp$, its BOMDP $\bomdp$, a belief-based strategy $\strat$
    for $\memdp$ and its lifted strategy $\liftedstrat$ for $\bomdp$, we have that the two induced MEMCs coincide: $\memdp[\strat] = \bomdp[\liftedstrat]$.
\end{restatable}\noindent

\noindent Consequently, satisfaction of objectives is preserved by the transformation. 
\begin{restatable}[]{theorem}{thmobjcorrespondence}\label{thm:obj_correspondence}
Let $\memdp$ be a MEMDP with state space $S$ and $\rabinObj$ a Rabin objective. 
Let $\liftedobj$ be the lifted Rabin objective to $S \times \powerset(I)$ by \Cref{def:lifted:rabin:objective}.
    A belief-based strategy $\strat$ for $\memdp$ is winning the Rabin objective $\rabinObj$ iff the lifted strategy $\liftedstrat$ is winning the lifted objective $\liftedobj$ for $\bomdp$:
$
    \forall \strat\colon \tuple{s, J} \models^{\memdp[\strat]} \rabinObj \Leftrightarrow \tuple{s, J} \models^{\bomdp[\liftedstrat]} \liftedobj
$.
\end{restatable}\noindent
As a result of \cref{thm:obj_correspondence}, we will implicitly lift strategies and objectives.

\subsection{An Algorithm for Reachability in BOMDPs}
\label{subsec:reachabilitybomdps}
\begin{algorithm}[t]
\begin{algorithmic}[1]
    \Function{Reach}{BOMDP $\bomdp$, $T \subseteq S$}
        \Do \label{alg:reach:outerloop}
            \For{$i \in I$} \label{alg:reach:innerloop}
                \State $S_i \assign \{ \tuple{s, J} \in S \times \powerset(I) \mid i \in J \}$
                \For{$\tuple{s, J} \in S_i \setminus \winregion{{\bomdp}_i}{\reachobj}$} \Comment{Iterate over all losing states}
                    \State $\bomdp \assign \stateremove{\bomdp}{\tuple{s, J}}$ \Comment{See \cref{def:stateremove}}
                \EndFor
            \EndFor
        \doWhile{$\bigwedge_{i \in I} S_i \ne \winregion{{\bomdp}_i}{\reachobj}$} \Comment{Check if stable}
        \State \Return $S^{\bomdp}$\label{line:reachreturn}
    \EndFunction
\end{algorithmic}
\caption{Reachability algorithm for a BOMDP $\bomdp$ of MEMDP $\memdp$.}%
\label{alg:reachability}%
\end{algorithm}%

In this subsection, we establish an algorithm for computing the winning region for reachability objectives in a BOMDP. 
The winning region of a BOMDP $\bomdp$ is precisely the set of winning beliefs of its MEMDP $\memdp$: $\winregion{\bomdp}{\Phi} = \{ \tuple{s,J} \in S \times \powerset(I) \mid \tuple{s,J} \models^{\memdp} \Phi \}$.
The algorithm specializes a similar fixed-point computation for POMDPS~\cite{DBLP:journals/jcss/ChatterjeeCT16} to BOMDPs.

\cref{alg:reachability} computes these winning regions. It relies on a state-remove operation defined below. 
Intuitively, the algorithm iteratively removes losing states, which does not affect the winning region until all states that remain in $\bomdp$ are winning.

Removing state $s$ from a BOMDP removes the state and disables outgoing action from any state where that action that could reach $s$ with positive probability.
This operation thus also removes any action and its transitions that could reach the designated state.

\begin{definition}[State removal]
\label{def:stateremove}
Let $\bomdp = \tuple{S \times \powerset(I), A, \initdist, {\{p_i\}}_{i \in I}}$ be a BOMDP, and $\perp \not\in S \times \powerset(I)$ a sink state.  
The BOMDP $\stateremove{\bomdp}{\tuple{s, J}}$ for $\bomdp$ and state $\tuple{s,J} \in S \times \powerset(I)$ is given by $\tuple{\{\perp\} \cup S \times \powerset(I) \setminus \{\tuple{s, J}\}, A, \initdist', {\{p_i'\}}_{i \in I}}$,
where $\initdist' = \perp$ if $\tuple{s,J} = \initdist$, and $\initdist'=\initdist$ otherwise, and
for all states $\tuple{s',J'} \neq \tuple{s,J}$ and environments $i \in I$ we have
\[\textstyle
p'_i(\tuple{s',J'},a) = \begin{cases}
    p_i(\tuple{s',J'},a) & \text{ if } \tuple{s,J} \not\in \supp(p_i(\tuple{s',J'}, a)),\\
    \dirac(\perp) & \text{ if } \tuple{s,J} \in \supp(p_i(\tuple{s',J'}, a)).
\end{cases}
\]
\end{definition}

\noindent 
The main results in this section are the correctness and the complexity of \cref{alg:reachability}:

\begin{restatable}[]{theorem}{thmreachabilitycorrect}\label{thm:reachability_correct}
    For BOMDP $\bomdp$ and targets $T$:  $\winregion{\bomdp}{\reachobj} = \algreach(\bomdp, T)$ in \cref{alg:reachability}.
\end{restatable}\noindent
Towards a proof, the notions of losing states and strategies as defined for MDPs also apply to BOMDP states and strategies.
For BOMDPs, we additionally define \emph{losing actions} as state-action pairs that lead with positive probability to a losing state.
It follows that a BOMDP state is losing iff every action from that state is losing, and a single environment where a BOMDP state is losing suffices as a witness that the state is losing in the BOMDP (see the \cref{sec:appendix:bomdps}). 
Finally, the following lemma is the key ingredient to the main theorem.
\begin{restatable}[]{lemma}{lemremovinglosingstates}\label{lem:removing_losing_states}
Removing losing states from $\bomdp$ does not affect the winning region, \ie,
$
\tuple{s, J} \not\in \winregion\bomdp\reachobj \text{ implies } \winregion{\stateremove{\bomdp}{\tuple{s, J}}}{\reachobj} = \winregion\bomdp\reachobj.
$
\end{restatable}%
\begin{restatable}[]{lemma}{lemreachabilitycomplexity}\label{lem:reachability_complexity}
    \cref{alg:reachability} takes polynomial time in the size of $\bomdp$.
\end{restatable}\noindent
\subsection{Safe \Bu in BOMDPs}\label{sec:safe_buchi}
In this section, we consider winning regions for \emph{safe \Bu objectives} of the form $\csafeobj \rabincobuchi \wedge \cbuchiobj \rabinbuchi$, where $\rabinbuchi \subseteq \rabincobuchi \subseteq S$.
The condition $\rabinbuchi \subseteq \rabincobuchi$ is convenient but does not restrict the expressivity. %
These objectives are essential for our Rabin algorithm in~\cref{sec:pspace:algorithm}. 
The main result is:
\begin{restatable}[]{theorem}{thmrabinpairpoly}\label{thm:rabin_pair_poly}
For BOMDP $\bomdp$, $\winregion\bomdp{\csafeobj \rabincobuchi \wedge \cbuchiobj \rabinbuchi}$ is computable in polynomial time. %
\end{restatable}\noindent
We provide the main ingredients for the proof below. We first consider arbitrary MEMDPs.
\begin{definition}[State restricted (ME)MDP]\label{def:state_restricted}
Let $\mdp = \tuple{\stateset, \actionset, \initdist, \transitionfunc}$ be an MDP and $S' \subseteq S$ a set of states. The MDP $\staterestrict{\mdp}{S'} \defeq \tuple{S' \cup \{\bot\}, \actionset, \initdist', \transitionfunc'}$ is $\mdp$ \emph{restricted to} $S'$, with $\bot$ a sink state,
$\initdist' = \initdist$ if $\initdist \in S'$ and $\bot$ otherwise, and  for $s \in S'$, $a \in A(s)$ and $s' \in S' \cup \{\bot\}$, we define:
\[\textstyle
\transitionfunc'(s, a, s') \defeq
\begin{cases}
\sum_{s'' \in S \setminus S'} \transitionfunc(s, a, s'') & \text{ if } s' = \bot,\\
\transitionfunc(s, a, s') & \text{ otherwise}.\\
\end{cases}
\]
This definition conservatively extends to MEMDPs per environment $i$: $({\staterestrict{\memdp}{S'}})_i = \staterestrict{(\memdp_i)}{S'}$.
\end{definition}

The winning regions of a MEMDP $\memdp$ and $\staterestrict\memdp\rabincobuchi$ coincide as, intuitively, winning strategies must remain in $\rabincobuchi$, thus removing other states does not affect the winning region.
\begin{restatable}[]{lemma}{lemsafebuechitobuechipreservationmemdp}  
\label{lem:safe_buechi_to_buechi_preservation_memdp}
The winning regions for $\csafeobj \rabincobuchi \wedge \cbuchiobj \rabinbuchi$ with $\rabinbuchi \subseteq \rabincobuchi$ in $\staterestrict\memdp\rabincobuchi$ and $\memdp$ coincide.
\end{restatable}\noindent
Satisfying the \Bu objective $\cbuchiobj\rabinbuchi$ inside  $\staterestrict\memdp\rabincobuchi$ implies satisfying the safety condition, thus:

\begin{restatable}[]{lemma}{lemsafebuechitobuechipreservationmemdptwo}
\label{lem:safe_buechi_to_buechi_preservation_memdp_two}
The winning regions for $\csafeobj \rabincobuchi \wedge \cbuchiobj \rabinbuchi$ with $\rabinbuchi \subseteq \rabincobuchi$ and $\cbuchiobj \rabinbuchi$ in $\staterestrict\memdp{\rabincobuchi}$  coincide. 
\end{restatable}\noindent
We can lift these lemmas to the BOMDP associated with a MEMDP.
\begin{restatable}[]{lemma}{lemsafebuchitobuchi}\label{lem:safe_buchi_to_buchi}
The winning regions for $\csafeobj \rabincobuchi \wedge \cbuchiobj \rabinbuchi$ with $\rabinbuchi \subseteq \rabincobuchi$ in $\bomdp$ and $\custombomdp{(\staterestrict{\memdp}{\rabincobuchi})}$ coincide.
\end{restatable}\noindent
Almost-sure \Bu objectives can be reduced to almost-sure reachability objectives using a construction similar to the one in~\cite{DBLP:journals/jacm/BaierGB12}, see the proof in the \cref{sec:appendix:bomdps} for details.
\begin{restatable}[]{lemma}{lembuchipoly}\label{lem:buchi_poly}
\Bu in BOMDPs is decidable in polynomial time.
\end{restatable}\noindent
Now, to prove \cref{thm:rabin_pair_poly}, $\winregion\bomdp{\csafeobj \rabincobuchi \wedge \cbuchiobj \rabinbuchi}$ is computable as \Bu objective on a polynomially larger MEMDP (\cref{lem:safe_buchi_to_buchi}) in polynomial time (\cref{lem:buchi_poly}).

\section{A Recursive PSPACE Algorithm for Rabin Objectives}\label{sec:pspace:algorithm}
We now show how to exploit the structure of BOMDPs to arrive at our $\tPSPACE$ algorithm for Rabin objectives in MEMDPs. 
We first discuss the \emph{non-local} behavior of Rabin objectives, and in particular, why the standard approach for almost-sure Rabin objectives for MDPs fails on BOMDPs.
Then, in \cref{sec:localview}, we introduce \emph{$J$-local MEMDPs}, which are MEMDPs where the belief $J$ does not change. 
These $J$-local MEMDPs also occur as fragments of the BOMDPs.
In  $J$-local MEMDPs, whenever a transition is made that would cause a belief update to a strict subset of $J$, we transition to dedicated sink states, which we refer to as \emph{frontier states}. 
These frontier states reflect transitioning into a different fragment of the BOMDP, from which all previously accessed BOMDP states are unreachable due to the monotonicity of the belief update operator.
Next, in \cref{sec:local_rabin_alg}, we present an algorithm for efficiently computing the winning region of Rabin objectives on $J$-local MEMDPs.
Finally, in \cref{sec:recursive_computation}, we prove that frontier states can be summarized as being either winning or losing, ultimately leading to a $\tPSPACE$ algorithm for deciding Rabin objectives in MEMDPs.

\subsection{Non-Local Behavior of Rabin Objectives}
\label{sec:nonlocalrabin}
The traditional approach for checking almost-sure Rabin objectives on MDPs, see \eg~\cite{Baier2008}, computes for each state $s \in S$, whether there is a strategy that \emph{immediately} satisfies a Rabin pair $\rabinObj_i = \rabinPair i$, \ie, satisfying  $\csafeobj\rabincobuchi_i \wedge \cbuchiobj\rabinbuchi_i$, and is a stronger condition.
A state satisfies the Rabin condition $\rabinObj$ iff it almost-surely reaches the set of immediately winning states (the \emph{win set}).
The example below illustrates why this approach fails to generalize to MEMDPs.
\begin{example}\label{ex:non_local}
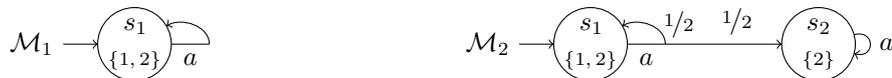
\begin{figure}[t]
    \caption{Example of a BOMDP fragment with Rabin objective $\rabinObj=\{\tuple{\{s_1\}, \{s_1\}}, \tuple{\{s_2\}, \{s_2\}}\}$.}
    \label{fig:non_local_example}
    \centering
    \begin{tikzpicture}[node distance=3cm, on grid, auto]
\begin{scope}
    \node[state, inner sep = 0, outer sep = 0, initial, initial text=$\memdp_1$, text width=0.7cm, text centered] (s1) {$s_1$ {\scriptsize$\{1, 2\}$}};

    \path[->]
        ($ (s1.east) + (0.5cm, 0) $)
            edge[-]             node{$a$} (s1)
            edge[out=90, in=30]           (s1)

        ;
\end{scope}

\begin{scope}[xshift=6cm]
    \node[state,inner sep = 0, outer sep = 0, initial, initial text=$\memdp_2$, text width=0.7cm, text centered] (s1) {$s_1$ {\scriptsize$\{1, 2\}$}};
    \node[state, inner sep = 0, outer sep = 0] (s2) [right=of s1, text width=0.7cm, text centered] {$s_2$ {\scriptsize$\{2\}$}};
    
    \path[->]
        ($ (s1.east) + (0.5cm, 0) $)
            edge[-]                   node{$a$}       (s1)
            edge[out=90, in=30, swap] node[right]{$\,\nicefrac12$} (s1)
            edge                      node[above]{$\quad\nicefrac12$} (s2)
            
        (s2)
            edge[loop right, looseness=3, swap] node{$a$}   (s2);
\end{scope}
\end{tikzpicture}
\end{figure}
In \cref{fig:non_local_example}, we see a BOMDP for which the `MDP approach' does \emph{not} work. First, note that the only strategy that always plays $a$ is winning in every state. 
Now, consider the algorithm and the first Rabin pair $\rabinObj_1 = \tuple{\{s_1\}, \{s_1\}}$. %
State $\tuple{s_2, \{2\}}$ does not satisfy $\csafeobj \{s_1\} \wedge \cbuchiobj \{s_1\}$.
State $\tuple{s_1, \{1, 2\}}$ also does not belong to the win set, as in $\memdp_2$ there is a $\nicefrac{1}{2}$ probability of reaching the sink state $\tuple{s_2, \{ 2 \}}$.
For the second Rabin pair, (only) state $\tuple{s_2, \{ 2 \}}$ is immediately winning.
Thus, the win set is the singleton set containing $\tuple{s_2, \{ 2 \}}$.
From the initial state $\tuple{s_1, \{1, 2\}}$, it is not possible to almost-surely reach the state $\tuple{s_2, \{2\}}$, due to $\memdp_1$.
Therefore, a straightforward adaption of the traditional algorithm for MDPs would yield that the initial state is losing.
\end{example}
The difficulty in the example above lies in the fact that in the different environments, a different Rabin pair is satisfied.
However, taking the self-loop in $s_1$ does not update the belief and it remains unclear whether we will eventually satisfy $\rabinObj_1$ or $\rabinObj_2$.

\subsection{Local View on BOMDPs}
\label{sec:localview}
We formalize $J$-local MEMDPs, that transition into frontier states if the belief updates.

\begin{definition}[$J$-local MEMDPs]
Given a MEMDP $\memdptuple$, the \emph{$J$-local MEMDP} $\localbomdp{\memdp}{J} = \tuple{S \sqcup F, A, \{p_j'\}_{j \in J}, \iota}$ is a MEMDP, with as state space the disjoint union of the (original) states $S$ and the \emph{frontier states} $F \defeq S \times A \times S$.
The transition functions $\{p_j' \colon S \sqcup F \times A \rightharpoonup \dist{S \sqcup F}\}_{j \in J}$ are defined
s.t.\ (1)~$p'_j(f,a,f) = 1$ for all $f \in F$, (2)~$p'_j(s,a)$ is undefined if $p_j(s,a)$ is undefined, and (3)~for every state $s\in S$ and $a \in A(s)$, we define $p'_j(s,a,\tuple{s,a,s'}) = p_j(s,a,s')$ if $\beliefupd(\tuple{s, J}, a, s') \neq \tuple{s', J}$ and  $p'_j(s,a,s') = p_j(s, a, s')$ otherwise.

\end{definition}
By definition of the transition functions $\{p_j'\}_{j \in J}$ of a $J$-local MEMDP $\localbomdp{\memdp}{J}$, all environments of $\localbomdp{\memdp}{J}$ share the same underlying graph within the states of $S$. Transitions to the frontiers may, however, differ (made formal in \cref{sec:appendix:local}).
\noindent
As both $\memdp$ and $\localbomdp{\memdp}{J}$ have states in $S$, a Rabin objective $\rabinObj$ can readily be applied to both.
To give meaning to the frontier states $F$ in $\localbomdp{\memdp}{J}$, we introduce \emph{localized Rabin objectives}:
\begin{definition}[Localized rabin objective, winning frontier]
Given Rabin objective $\rabinObjDef$ and some subset of frontier state $\winningfrontier \subseteq F$, the \emph{localized Rabin objective} for $J$-local MEMDP $\localbomdp{\memdp}{J}$ is $
\locRabinObj(\winningfrontier) \defeq \{ \tuple{\rabincobuchi_i \cup \winningfrontier, \rabinbuchi_i \cup \winningfrontier} \mid 1 \le i \le k \}$.
We call $\winningfrontier$ the \emph{winning frontier}, as any path that reaches a state in $\winningfrontier$ is winning.
\end{definition}

\subsection{An Algorithm for Localized Rabin Objectives}\label{sec:local_rabin_alg}
Below, we present an algorithm to compute the winning region of a localized Rabin objective on a $J$-local MEMDP,
using some auxiliary definitions on winning in a $J$-local MEMDPs.
\begin{definition}[Immediately winning Rabin pair/state]
A $J$-local MEMDP state $s \in S \sqcup F$ has an \emph{immediately winning Rabin pair} $\rabinObj_i = \rabinPair{i}$ when $s \models^{\localbomdp{\memdp}{J}} \csafeobj\rabincobuchi_i \wedge \cbuchiobj\rabinbuchi_i$.
A state $s \in S \sqcup F$ is \emph{immediately winning} if it has an immediately winning Rabin pair.
\end{definition}
Immediately winning states are, in particular, also winning states (see \cref{lem:immediately_winning_is_winning}, \cref{appx:algorithm:loc:rabin}).
\noindent
It is natural also to consider specialized winning regions for just immediately winning states:
\begin{definition}%
The \emph{Rabin win set} $\rabinWinSet$ is $\{ s \in S \sqcup F \mid s \text{ is immediately winning } \}$.
\end{definition}
The crux of our algorithm is that in $J$-local MEMDPs, as in MDPs but unlike in BOMDPs, winning a Rabin objective is equivalent to almost-surely reaching the Rabin win set.
\begin{restatable}[]{lemma}{lemcommonpairreachable}\label{lem:common_pair_reachable}   
    A state $s$ in a $J$-local MEMDP is winning iff it can almost-surely reach $\rabinWinSet$.
\end{restatable}
\noindent
We sketch the proof ingredients later.
We first introduce~\cref{alg:local_rabin}, which lifts the MDP approach (\cref{sec:nonlocalrabin}) to $J$-local MEMDPs.
The set $\winset$ on line~\ref{line:winset} stores states for which an immediately winning Rabin pair has been found.
For each Rabin pair $\rabinObj_i$, the algorithm computes the localized Rabin pair $\locRabinObj_i$.
Next, in line~\ref{line:safebuchi}, it compute the winning region $\winregion{ \genericlocal }{ \csafeobj{\rabincobuchi_i'} \wedge \cbuchiobj{\rabinbuchi_i'} }$ using the approach described in~\cref{sec:safe_buchi}.
These are exactly the states that have $\locRabinObj_i$ as an immediately winning Rabin pair, \ie, they constitute the win set  $\winset$.
Finally, the algorithm outputs the winning region by computing states that almost-surely reach $\winset$ using \cref{alg:reachability}.

\begin{algorithm}[t]
\begin{algorithmic}[1]
  \Function{\rabinalgname}{Local MEMDP $\genericlocal = \localbomdp{\memdp}{J}, \winningfrontier, \rabinObj = \{ \tuple{\rabincobuchi_1, \rabinbuchi_1}, \cdots, \tuple{\rabincobuchi_n, \rabinbuchi_n}\}$ }

  \State $\winset \assign \emptyset$\label{line:winset}
  \For{$1 \le i \le n$}\label{line:forloop}
    \State $\rabincobuchi_i' \assign \rabincobuchi_i \cup \winningfrontier$ \, ; \, $\rabinbuchi_i' \assign \rabinbuchi_i \cup \winningfrontier$

    \State $\winset \assign \winset \cup \winregion{ \genericlocal }{ \csafeobj{\rabincobuchi_i'} \wedge \cbuchiobj{\rabinbuchi_i'} }$
    \Comment{See~\cref{thm:rabin_pair_poly}}\label{line:safebuchi}
  \EndFor

  \State \Return{$\winregion{ \genericlocal }{ \creachobj{\winset} }$}\label{line:return}
  
  \EndFunction
\end{algorithmic}
\caption{Local Rabin Algorithm.}%
\label{alg:local_rabin}%
\end{algorithm}%

\begin{restatable}[]{theorem}{thmrabinalgcorrect}\label{thm:rabin_alg_correct}
\Cref{alg:local_rabin}
yields winning regions for local MEMDPs and localized Rabin objectives.
\end{restatable}
\noindent
The remainder of this subsection discusses the ingredients for proving \cref{lem:common_pair_reachable} and the theorem above. 
Therefore, we consider the induced Markov chain $\mc$ of environment $j$ under any strategy, \ie, $\mc = \localbomdp\memdp J[\strat]_j$. 
In any state that is in a BSCC of $\mc$, we notice that the reachable states \emph{in any environment} are contained by the BSCC and the frontier states. 
Furthermore, we observe that in any environment, either the BSCCs in those states are the original BSCC or are (trivial) BSCCs in the frontier.
Formal statements are given in \Cref{appx:algorithm:loc:rabin}.
The next lemma shows that states that are (under a winning strategy and in some environment) in a BSCC are immediately winning with some Rabin pair. 
The main challenge is that this BSCC may not be a BSCC in every environment. 
Using the observations above, if the states do not constitute a BSCC, they will almost surely reach (winning) frontier states, which allows us to derive the following formal statement:
\begin{restatable}[]{lemma}{lemcommonrabinpair}\label{lem:common_rabin_pair}
Given a $J$-local MEMDP $\localbomdp{\memdp}{J}$ and a winning strategy $\strat$.  
Every state that is in a BSCC 
$\Sbsccj j$ of $\localbomdp{\memdp}{J}[\strat]_j$
of some environment $j \in J$,
is in $\rabinWinSet$.
\end{restatable}\noindent
With this statement, we can now prove \cref{lem:common_pair_reachable} as under any winning strategy, we almost-surely end up in BSCCs.
We return to the proof of the main theorem about the correctness of \rabinalgname{}.
First, we observe that we correctly identify the immediately winning states.
\begin{restatable}[]{lemma}{lemwinsetcorrect}\label{lem:winset_correct}
   \rabinalgname{}  computes the set of states that are immediately winning, $\rabinWinSet$.
\end{restatable}\noindent
\cref{lem:common_pair_reachable,lem:winset_correct} together prove~\cref{thm:rabin_alg_correct}. Finally, we remark:
\begin{restatable}[]{lemma}{lemlocalrabinpoly}
\label{lem:local_rabin_poly}
\rabinalgname{} (\Cref{alg:local_rabin}) is a polynomial time algorithm.
\end{restatable}\noindent

\subsection{Recursive Computation of Winning Regions}\label{sec:recursive_computation}
We now detail how to combine the local computations of winning regions towards a global winning region.
Furthermore, we show that to obtain the winning region at the root (\ie, $I$-local), we can forget about the winning regions below and, consequently, present a recursive approach (akin to \cite{DBLP:conf/tacas/VegtJJ23}) to decide almost-sure Rabin objectives for MEMDPs in $\tPSPACE$.

\begin{restatable}[]{theorem}{thmrabinpspace}\label{thm:rabin_pspace}
Winning almost-sure Rabin objectives in MEMDPs is decidable in $\tPSPACE$.
\end{restatable}\noindent
In the remainder, we show this by providing a recursive algorithm and proving its correctness. 
An important construction is to project the winning region into a particular set of beliefs.
\begin{definition}[Belief-restricted winning regions]\label{def:restricted_winning_regions}
For a Rabin objective $\rabinObj$, we define the following restrictions of the winning region: (1)~$\winregion{\memdp}{\rabinObj}_J \defeq \winregion{\memdp}{\rabinObj} \cap (S \times \{J\})$, (2)~$\winregion{\memdp}{\rabinObj}_{\subset J} \defeq \bigcup_{J' \subset J} \winregion{\memdp}{\rabinObj}_{J'}$, and (3)~$\winregion{\memdp}{\rabinObj}_{\subseteq J} \defeq \winregion{\memdp}{\rabinObj}_{\subset J} \cup \winregion{\memdp}{\rabinObj}_J$.
\end{definition}
We now define the localized Rabin objective where we determine the winning frontiers based on the actual winning states in a BOMDP. We use the following auxiliary notation: We define the reachable frontier $\mathit{RF}\defeq\reachable{S} \cap F$. Then, we can determine where a local transition $s\xrightarrow{a} s'$ leads in the global system, $ \mathit{ToGlob}_J(\tuple{s, a, s'}) \defeq \beliefupd(\tuple{s, J}, a, s')$ and finally consider $\mathit{WinLocal}_J(F, B) \defeq \{ f \in F \mid \mathit{ToGlob}_J(f) \in B \}$.\sj{wip}
\begin{definition}[Correct localized Rabin objective]
For belief $J$, the \emph{correct localized Rabin objective} is $\clocRabinObj J \defeq \locRabinObj(\mathit{WinLocal}_J(\mathit{RF}, \winregion{\memdp}{\rabinObj}_{\subset J}))$.
\end{definition}
The notion of correctness in the definition above is justified by the following theorem, which says that computing the correct localized Rabin objective provides the belief-restricted winning region. 
That is, the winning region of the $J$-local MEMDP $\localbomdp{\memdp}{J}$ with its correct localized Rabin objective is equal to the global winning region restricted to $J$.
\begin{restatable}[]%
{theorem}{thmlocalwinningmeansglobal}\label{thm:local_winning_means_global}%
For Rabin objective $\rabinObj$:
$(\winregion{\localbomdp{\memdp}{J}}{\clocRabinObj J} \cap S) \times \{J\} = \winregion{\memdp}{\rabinObj}_{J}$.
\end{restatable}\noindent
The theorem immediately leads to the following characterization of the winning region.
\begin{corollary}\label{cor:memdp_win_region_characterisation}
For Rabin objective $\rabinObj$:
$\winregion{\memdp}{\rabinObj} = \bigcup_{J} (\winregion{\localbomdp{\memdp}{J}}{\clocRabinObj J} \cap S) \times \{J\} $.
\end{corollary}\noindent
\begin{algorithm}[t]
\begin{algorithmic}[1]
  \Function{\genericalgoname}{MEMDP $\memdptuple, \rabinObj$ }
    \State $\genericlocal \assign \localbomdp{\memdp}{I}$
    \State $\mathit{RF} \assign \reachable{S^{\genericlocal}} \cap F^\genericlocal$ \Comment{Compute the reachable frontier states}
    \State $\winningfrontier \assign \{ \tuple{s, a, s'} \in \mathit{RF} \mid \tuple{s', J'}=\beliefupd(\tuple{s, J}, a, s') \wedge \genericalgoname(\restrictenv{\newinit{\memdp}{{s'}}}{J'}, \rabinObj) \}$\label{line:wf}
    \State \Return{$\initdist \in \rabinalgname{}(\genericlocal, \winningfrontier, \rabinObj)$} \Comment{Compute winning set with winning frontier}\label{line:returngeneric}
  \EndFunction
\end{algorithmic}
\caption{Generic recursive algorithm for MEMDPs.}%
\label{alg:genericalgorithm}%
\end{algorithm}%
\Cref{cor:memdp_win_region_characterisation} suggests computing the winning region from local MEMDPs. 
The computation can go bottom-up, as the winning region of MEMDP restricted to a belief $J$ only depends on the $J$-local MEMDP $\localbomdp\memdp J$ and the winning regions of beliefs $J' \subset J$.
These observations lead us to \Cref{alg:genericalgorithm}.
We construct the $J$-local MEMDP, recursively determine the winning status of all its frontier states, and
then compute the local winning region of $\localbomdp\memdp J$.
\begin{restatable}[]{theorem}{thmcheckcorrect}
\label{thm:check_correct}
In \Cref{alg:genericalgorithm} with Rabin objective $\rabinObj$:  $\genericalgoname(\memdp, \rabinObj) \text{ iff } \initdist \in \winregion{\memdp}{\rabinObj}.$
\end{restatable}%
\begin{restatable}[]{lemma}{thmcheckspace}\label{thm:check_space}
\Cref{alg:genericalgorithm}  runs in polynomial space.
\end{restatable}\noindent
This lemma follows from observing that a local MEMDP and thus its frontier is polynomial and that the recursion depth is limited by $|I|$.
\cref{thm:check_correct,thm:check_space} together prove the main theorem \cref{thm:rabin_pspace}: The decision problem of almost-sure Rabin objectives in MEMDPs is in $\tPSPACE$.
Thus, almost-sure safety, \Bu, \CoB, and parity are in $\tPSPACE$ too~\cite{DBLP:conf/icalp/ChatterjeeAH05}.

\begin{restatable}[]{theorem}{thmalgbigo}\label{thm:alg_big_o}
The time complexity of \Cref{alg:genericalgorithm}  is in $O((|S|^2\cdot|A|)^{|I|} \cdot \mathit{poly}(|\memdp|, |\rabinObj|))$.
\end{restatable}\noindent
The bound in \Cref{thm:alg_big_o} is conservative\footnote{A more precise bound can likely be obtained from the number of revealing transitions in the MEMDP.}, and it shows that deciding almost-sure Rabin objectives for $2$-MEMDPs is in $\tP$.
Almost-sure parity objectives for $2$-MEMDPs were already known to be in $\tP$~\cite{DBLP:conf/fsttcs/RaskinS14}. Indeed, it establishes the complexity for any fixed number of constants\footnote{That is, the decidability problem is in $\tXP$ with parameter \emph{number of environments} $k$.}.

\begin{restatable}[]{corollary}{cortwomemdps}\label{cor:2memdps}
For constant $k$, deciding almost-sure Rabin for $k$-MEMDPs is in $\tP$.
\end{restatable}\noindent

\section{Conclusion}
We have presented a $\tPSPACE$ algorithm for almost-sure Rabin objectives in MEMDPs.
This result establishes $\tPSPACE$-completeness for many other almost-sure objectives, including parity, and completes the complexity landscape for MEMDPs.
We additionally showed that all objectives under the possible semantics we consider in MEMDPs belong to the same complexity classes as MDPs.
Interesting directions for future work are to investigate whether the constructions used in this paper can also be of benefit for quantitative objectives in MEMDPs or more expressive subclasses of POMDPs, for example, a form of MEMDPs where the environments may change over time.

\newpage
\bibliography{refs}

\newpage
\appendix

\appendix

{\section*{Appendix}}
\noindent
This appendix contains all the proofs, auxiliary definitions, and lemmas needed.
We follow the same structure as the main paper. 
The proof of each statement in the main text can be found in its respective (sub)section of this appendix.
For convenience, we restate each statement from the main text before its proof.
\medskip

\section{Auxiliary Material of \Cref{sec:background}}\label{appx:aux:sec:background}

\subsubsection*{Objectives}
We formally define reachability, safety, \Bu, \CoB, and parity objectives.
Let $T \subseteq \stateset$ be some set of target states, $\Path \in \paths(\mc)$ a path, and $\infset(\Path) \subseteq S$ again the set of states reached infinitely often following path $\Path$. 
Finally, let $l \colon \stateset \to \mathbb{N}$ be a priority function.
We use LTL-like notation as short-hand for most of the objectives.
\begin{description}
    \item[Reachability.] By $\reachobj$ we denote the objective to eventually reach a target set $T \subseteq S$. 
    The Markov chain $\mc$ satisfies $\reachobj$ almost-surely if and only if $\mathbb{P}_{\mc}(\reachobj) = 1$ (or possibly when $> 0$).
    \item[Safety.] By $\safeobj$ we denote the objective to remain within the a set of safe states $T \subseteq S$. 
    $\mc$ satisfies $\safeobj$ almost-surely if and only if $\mathbb{P}_{\mc}(\safeobj) = 1$ (or possibly when $> 0$).
    \item[\Bu.] By $\buchiobj$ we denote \Bu objectives. $\mc$ satisfies a \Bu objective $\buchiobj$ almost-surely if and only if 
    $\mathbb{P}_{\mc}(\Path \in \paths(\mc) \mid T \cap \infset(\Path) \ne \emptyset) = 1$ (or possibly when $> 0$).
    \item[\CoB.] By $\cobuchiobj$ we denote \CoB objectives. $\mc$ satisfies a \CoB objective $\cobuchiobj$ almost-surely if and only if $\mathbb{P}_{\mc}(\Path \in \paths(\mc) \mid \infset(\Path) \subseteq T) = 1$ (or possibly when $> 0$).
    \item[Parity.] We denote parity objectives by their priority function $l$. 
    $\mc$ satisfies the parity objective $l$ almost-surely if and only if $\mathbb{P}_{\mc}(\Path \in \paths(\mc) \mid \max\{l(\infset(\Path))\} \equiv 0 \pmod2) = 1$ (or possibly when $> 0$).
\end{description}

\section{Proofs and Auxiliary Material of \cref{subsec:memdps}}

\paragraph*{Possible Objectives}

\begin{lemma}\label{lemma:memdp:possible:winning}
    Given a MEMDP $\memdp = \{M_i\}_{i\in I}$ and possible Rabin objective $\Phi$, the MEMDP is winning iff all environments $M_i$ are winning.
\end{lemma}
\begin{proof}
    Assume $\memdp$ is winning. 
    Then there exists a strategy $\strat$ such that every individual environment $M_i$ is winning under $\strat$, by definition of winning in MEMDPs (see \Cref{subsec:memdps}).

    Assume all environments $\{M_i\}_{i \in I}$ are winning. 
    Then there exist winning strategies, one for each environment, $\{\strat_i\}_{i \in I}$.
    Since memoryless deterministic strategies suffice for possible Rabin objectives in MDPs~\cite{DBLP:conf/icalp/ChatterjeeAH05}, w.l.o.g., all these strategies are of type $\strat_i \colon S \to A$.
    From these strategies, a new winning strategy for $\memdp$ can be constructed by uniformly randomizing over these strategies:
    \[
        \strat_{\memdp}(s) = \unif(\{\strat_i(s)\}_{i \in I}).
    \]
    As every strategy $\strat_i$ is possibly winning, and randomizing over their winning actions ensures that the probability of winning remains greater than zero.
    Hence, $\strat_{\memdp}$ is possibly winning.
\end{proof}

\thmposequaltomdps*
\begin{proof}
    Let $\memdp = \{M_i\}_{i \in I}$ be a MEMDP.
    For each environment $M_i$, check whether $M_i \models \Phi$, sequentially.
    By \Cref{lemma:memdp:possible:winning}, we know that if every environment $M_i$ is winning, the MEMDP is also winning.
    Hence, we only need to keep track of a single counter to keep track of the environment and a flag to store whether the $M_i$ are winning. Thus, the complexity is dominated by the complexity of the subroutine used to check $M_i \models \Phi$, which is in $\tNL$ for reachability and in $\tP$ for the other objectives~\cite{DBLP:conf/mfcs/ChatterjeeDH10}.
\end{proof}

\section{Proofs and Auxiliary Material of \cref{sec:belief:sufficient}}

\subsubsection*{4.2 Belief-Based Strategies}

\lembeliefstratfsc*
\begin{proof}[Proof \Cref{lem:beliefstrat:fsc}]
We construct the following FSC.
Let $\fsc = \tuple{\powerset(I), J_{\iota}, \alpha, \eta}$ be an FSC with memory nodes $\powerset(I)$, initial memory node $J_\iota = I$, the strategy $\strat$ as action mapping, \ie, $\alpha = \strat$, and memory update $\eta$ encodes the belief update function by
$
    \eta(J,a,s') = \dirac(J') \Leftrightarrow \tuple{s',J'} = \beliefupd(\tuple{s,J},a,s').
$
\end{proof}

\subsubsection*{4.3 Constructing Belief-Based Strategies}\label{appx:constructing:belief:based:strategies}

Recall the definition of allowed actions (\Cref{def:allowed:actions}) and the strategy that always remains winning $\stratallow \colon S \times \powerset(I) \to \dist{A}$:
\[
\stratallow(\tuple{s, J}) \defeq \begin{cases}
    \unif(\allow(\tuple{s,J})) & \text{ if } \allow(\tuple{s,J}) \neq \emptyset, \\
    \unif(A(s)) & \text{ otherwise.}
\end{cases}
\]
The strategy $\stratallow$ randomizes uniformly over all allowed actions when the successor beliefs are still winning and overall actions when the belief cannot be winning.

We now examine the behavior of $\stratallow$ on the MEMDP $\memdp$ by looking at the induced MC for each individual environment and the induced MEMC for the MEMDP $\memdp$ as a whole.

Consider the induced MEMC $\memdp[\stratallow] = \{\memdp[\stratallow]_i\}_{i \in I}$.
Since $\stratallow$ is a belief-based strategy, which can be represented by an FSC (\Cref{lem:beliefstrat:fsc}), the state space of the MEMC $\memdp[\stratallow]$ is given by the set of all possible beliefs $S \times \powerset(I)$.
Given a belief $\tuple{s,J}$, we can compute all the environments $j \in J$ for which $s$ is part of a BSCC $\Sbscc \subseteq S \times \powerset(I)$ in MC $\memdp[\stratallow]_j$ as %
\[
\Jbscc = \{ j \in J \mid \exists \Sbscc \subseteq S \times \powerset(I) \colon \Sbscc \text{ is a BSCC in MC } \memdp[\stratallow]_j \wedge \tuple{s,J} \in \Sbscc \}.
\]
An important observation is that for any belief-based strategy $\strat$, and for all the environments $j \in \Jbscc$, the MCs $\memdp_j[\strat]$ 
starting in state $\tuple{s,J}$ have the same underlying graph:

\begin{restatable}[]{lemma}{lemmajbsccgraph}\label{lemma:Jbscc:graph}
Let $\strat$ be a belief-based strategy, and $\tuple{s,J}$ a belief. 
We define the \emph{induced MC for environment $j \in J$ starting in $\tuple{s,J}$} as
$\newinit{\memdp}{{\tuple{s,J}}}_j[\strat] = \tuple{\reachablem{\memdp_j}{s} \times \powerset(J), \tuple{s,J}, p_j'}$, 
where $p_j'(\tuple{s,J},\tuple{s',J'})$ is defined as in the induced MEMC of a MEMDP and FSC.
Then:  
\[ 
\forall j,j' \in \Jbscc, \text{ the MCs } \newinit{\memdp}{{\tuple{s,J}}}_j[\strat] \text{ and } \newinit{\memdp}{{\tuple{s,J}}}_{j'}[\strat]\text{ have the same underlying graph.}
\] 
\end{restatable}\noindent 
\begin{proof}[Proof \Cref{lemma:Jbscc:graph}]
 Proof by contradiction. 
    We assume two environments $j, j' \in \Jbscc$ such that $\newinit{\memdp}{{\tuple{s,J}}}_j[\strat]$ and $\newinit{\memdp}{{\tuple{s,J}}_{j'}}[\strat]$ do not have the same underlying graph. 
    Without loss of generality, assume environment $j$ has a transition $(\tuple{s,J},\tuple{s',J'})$ that environment $j'$ does not have. 
    Hence, it is a revealing transition.
    From the monotonic information gain property (\Cref{lemma:belief-stabalizing}) it then follows that $J' \subset J$, which means the state $\tuple{s,J}$ was, in fact, not part of a BSCC $\Sbscc$, and thus $j \notin \Jbscc$, contradicting our assumption.
\end{proof}

As a consequence, since $\stratallow$ is a belief-based strategy, every BSCC of the MEMC $\memdp[\stratallow]$ has a fixed set of environments that cannot change anymore. We call this property \emph{belief stabilization}:
\begin{corollary}[Belief stabilization]\label{col:belief:stabilization}
Given a BSCC $\Sbscc$ of $\memdp[\stratallow]$, for all states $\tuple{s,J}, \tuple{s',J'} \in \Sbscc$, we have $J = J'$.
\end{corollary}
Now consider the sub-MDPs $\{\tuple{\Sbscc, A, \initdist', \transitionfunc_j'}\}_{j \in \Jbscc}$ of the MDPs $\{\memdp_j\}_{j \in \Jbscc}$, 
with initial state $\initdist' = \tuple{s, J}$, and transition function $\transitionfunc_j' \colon \Sbscc \times A \rightharpoonup \dist{\Sbscc}$ being 
\[
\transitionfunc_j'(\tuple{s,J},a, \tuple{s',J'}) \defeq
\begin{cases}
    \transitionfunc_j(s,a,s') & \text{ if } a \in \allow(\tuple{s,J}),\\
    \text{undefined} & \text{ otherwise},
\end{cases}
\] 
restricted to BSCC states and allowed actions.

\begin{restatable}[]{lemma}{lemmasubmdpgraph}\label{lemma:submdp:graph}
The sub-MDPs $\{\mdp_j = \tuple{\Sbscc, A,\tuple{s,J}, \transitionfunc_j'}\}_{j \in \Jbscc}$ have the same underlying graph.
\end{restatable}\noindent
\begin{proof}[Proof \Cref{lemma:submdp:graph}]
 Assume that two sub-MDPs $M_j, M_{j'}$, consisting of BSCC states $\Sbscc$ and allowed actions $\Aallow$, do not have the same underlying graph.
 Without loss of generality, assume $M_j$ contains a transition $(\tuple{s,J}, a, \tuple{s',J'})$ that is not present in $M_{j'}$.
 By construction, the states $\tuple{s,J}, \tuple{s',J'} \in \Sbscc$ and $a \in \Aallow$.
 Hence, $a$ is an allowed action, and thus $\stratallow(\tuple{s,J})(a) > 0$.
 Then the induced MC $M_j[\stratallow]$ also contains a transition $(\tuple{s,J}, \tuple{s',J'})$ that $M_{j'}[\stratallow]$ does not have.
 Since $M_j$ and $M_{j'}$ both consist of states that form a BSCC under $\stratallow$, this missing transition in $M_{j'}[\stratallow]$ is a direct contradiction with \Cref{lemma:Jbscc:graph}.
\end{proof}

\begin{restatable}[]{corollary}{corstratbsccwinning}\label{cor:stratbscc:winning}
    Given a sub-MDP $\mdp_j = \tuple{\Sbscc, A, \tuple{s,J}, \transitionfunc_j'}$ for any $j \in \Jbscc$, we can compute a strategy $\stratbscc \colon \Sbscc \to \dist{A}$ that is almost-surely winning for objective $\Phi$ on $M_j$.
    Then $\stratbscc$ is also winning for all sub-MDPs $\{M_{j}\}_{j \in \Jbscc}$.
\end{restatable}\noindent
\begin{proof}[Proof \Cref{cor:stratbscc:winning}]
    By \Cref{lemma:submdp:graph} we know that all sub-MDPs $\{\mdp_j\}_{j \in \Jbscc}$ have the same underlying graph and have finite sets of states and actions.
    Almost-sure Rabin objectives $\Phi$ are qualitative objectives that only depend on the graph to win in an MDP~\cite[Section 10.6.3]{Baier2008}~\cite{DBLP:conf/soda/ChatterjeeJH04,DBLP:conf/icalp/ChatterjeeAH05}.
    Thus, given two finite MDPs $M_1, M_2$ that share the same underlying graph, we have that there exists a strategy $\strat$ such that $M_1[\strat]$ is winning iff $M_2[\strat]$ is winning~\cite{DBLP:conf/soda/ChatterjeeJH04}.
\end{proof}

\thmbeliefsufficient*
\begin{proof}[Proof \cref{thm:beliefsufficient}]
Given a winning belief-based strategy, the existence of a (general) winning strategy is trivial, as belief-based strategies are a subset of all strategies since beliefs are based on paths.

For the other direction, we construct a belief-based strategy that is winning for objective $\Phi$ as follows
\[
\stratwin(\tuple{s,J}) \defeq \begin{cases}
    \stratallow(\tuple{s,J}) & \text{ if } \Jbscc = \emptyset,\\
    \stratbscc(\tuple{s,J}) & \text{ otherwise.}  
\end{cases}
\]
We now show that $\stratwin$ is indeed winning.
For every belief, including the initial belief, we check whether the belief $\tuple{s,J}$ contains at least one environment $j \in J$ where $s$ is in a BSCC in MC $\memdp_j[\stratallow]$.
That is, $\Jbscc \neq \emptyset$.
When that is not the case, and $\Jbscc = \emptyset$, we play allowed actions according to $\stratallow$.

When $\emptyset \neq \Jbscc \subseteq J$, %
we use a winning strategy $\stratbscc$ for any environment $j \in \Jbscc$.
By \cref{cor:stratbscc:winning}
this strategy is winning for all environments $j \in \Jbscc$.
If $\Jbscc = J$, it follows directly that $\stratwin$ is winning for the whole MEMDP.
From \Cref{col:belief:stabilization}, it immediately follows that when $\Jbscc \subset J$, we were, in fact, not in a BSCC, and there exists a revealing transition witnessing this fact and leading to a belief update from $\tuple{s,J}$ to $\tuple{s',J'}$.
If we never witness this revealing transition under $\stratwin$, then $\stratwin$ was already a winning strategy for our Rabin objective, and if we do witness the revealing transition, the belief is updated, and we are back at checking whether $\mathcal{J}_{\tuple{s',J'}}$ contains an environment where $s'$ is in a BSCC.
Due to the monotonic information gain property (\Cref{col:belief:stabilization}), this process terminates, and
hence, $\stratwin$ is indeed winning for MEMDPs with Rabin objectives.%
\end{proof}

\section{Proofs and Auxiliary Material of \cref{sec:bomdps}}\label{sec:appendix:bomdps}

\subsubsection*{5.1 Belief-Observation MDPs}
\lemstratcorrespondence*
\begin{proof}[Proof of \cref{lem:strat_correspondence}]
Let $\memdp = \tuple{\stateset, \actionset, \initdist, \transitionfuncs}$ and $\bomdp = \tuple{S \times \powerset(I), \actionset, \initdist', \{p'_i\}_{i \in i}}$.
We use \cref{lem:beliefstrat:fsc} to turn $\strat$ into an FSC to compute the induced MEMC. 
The FSC has the form $\fsc = \tuple{\powerset(I), J_{\iota}, \strat, \eta}$, so the induced MEMC $\memdp[\fsc]$ has the form $\tuple{S \times \powerset(I), \initdist', \{ \transitionfunc^a_i \}_{i \in I}}$ with 
$\initdist' = \tuple{\initdist, I}$ and the transition function $\transitionfunc^a$ by
\begin{align*}
\transitionfunc^a_i(\tuple{s, J},\tuple{s',J'}) &= \sum_{a \in A} \strat(\tuple{s,J})(a) \cdot \transitionfunc_i(s,a,s') \cdot \eta(J,a,s',J')\\
&= \begin{cases}
    \sum_{a \in A} \strat(\tuple{s,J})(a) \cdot \transitionfunc_i(s, a,s') & \text{ if } \tuple{s', J'}=\beliefupd(\tuple{s, J}, a, s'),\\
    0 & \text{ otherwise}.
\end{cases}%
\end{align*}%
By definition, it is clear that $\bomdp[\liftedstrat]$ has the same state space and initial distribution.
Applying $\liftedstrat$ as a memoryless strategy (using an FSC with a single memory node) yields
\begin{align*}
    \transitionfunc^b_i(\tuple{s, J},\tuple{s',J'}) =&
        \sum_{a \in A} \strat(\tuple{s, J})(a) \cdot \transitionfunc'_i(\tuple{s, J}, a,\tuple{s', J'})\\
        =&\ \begin{cases}
             \sum_{a \in A} \strat(\tuple{s, J})(a) \cdot \transitionfunc_i(s, a,s') & \text{ if } \tuple{s',J'} = \beliefupd(\tuple{s, J}, a, s'),\\
             0 & \text{ otherwise},
        \end{cases}
\end{align*}
which is equal to $\transitionfunc^a$.
\end{proof}

\thmobjcorrespondence*
\begin{proof}[Proof of \cref{thm:obj_correspondence}]
By~\cref{lem:strat_correspondence}, we know that $\memdp[\strat] = \bomdp[\liftedstrat]$.
As the memory of the FSC representing belief-based policy $\strat$ is equal to $\powerset(I)$~(\cref{lem:beliefstrat:fsc}), the FSC lifted objective as defined in~\cref{def:lifted:rabin:objective} is equal to the BOMDP lifted objective $\liftedobj$.
\end{proof}

\subsubsection*{5.2 An Algorithm for Reachability in BOMDPs}
We first give a formal definition of losing actions.
\begin{definition}[Losing action]
We say that in a BOMDP an action $a \in A(\tuple{s, J})$ is losing in state $\tuple{s, J}$ when $(\tuple{s,J}, a)$ has a losing successor state in some environment $j \in J$:
\[
\supp(p_j(\tuple{s, J}, a)) \cap ((S \times \powerset(I)) \setminus \winregion\bomdp\reachobj) \ne \emptyset.
\]
\end{definition}

We start with two additional lemmas.
\begin{lemma}\label{lem:losing_state_action}
A BOMDP state $\tuple{s, J} \not\in T$ is losing iff every action $a \in A(\tuple{s, J})$ is a losing action.
\end{lemma}

\begin{proof}[Proof of \cref{lem:losing_state_action}]
First, we prove that if every action $a \in A(\tuple{s, J})$ is losing, then necessarily $\tuple{s, J}$ is losing.
Assume for contradiction that $\tuple{s, J} \in \winregion\bomdp\reachobj$, then there exists a winning strategy $\strat$.
Thus, $\strat(\tuple{s, J})(a) > 0$ for some $a \in A(\tuple{s, J})$, which is a losing action.
Therefore, there exists an environment $j \in J$ such that we can reach a losing state $\tuple{s', J'}$.
As $\tuple{s', J'} \not\in \winregion\bomdp\reachobj$, expanding the definition of losing, we know that $\forall\strat'\colon \exists {j'} \in J'\colon \mathbb{P}_{{\bomdp}_j[\strat']}(\reachobj \mid s_0 = s') < 1$.
As $\tuple{s', J'}$ is reached with positive probability from $\tuple{s, J}$, we have that $\mathbb{P}_{{\bomdp}_j[\strat']}(\reachobj \mid s_0 = s) < 1$.
This contradicts the assumption that $\strat$ is winning.

As for the other direction of the proof, we need to show that if a state $\tuple{s, J}$ is losing, then necessarily every action $a \in A(\tuple{s, J})$ is losing. 
Suppose for contradiction that there exists an action $a \in A(\tuple{s,J})$ that is not losing. 
By definition of losing actions, the action is thus winning.
Then, no losing state can be reached in any environment $j \in J$. 
Hence, all states reachable by playing $a$ in $\tuple{s,J}$ are winning in all environments $j \in J$.
Hence, the current state $\tuple{s,J}$ also has to be winning. 
A contradiction.
\end{proof}

\begin{lemma}\label{lem:losing_inclusion}
If there exists an environment $j \in J$ such that $\tuple{s, J} \not\in \winregion{{\bomdp}_j}\reachobj$, then $\tuple{s,J} \not\in \winregion\bomdp\reachobj$.
\end{lemma}
\begin{proof}[Proof of \Cref{lem:losing_inclusion}]
    For contradiction, assume that $\tuple{s, J} \not\in \winregion{{\bomdp}_j}\reachobj$, but $\tuple{s,J} \in \winregion\bomdp\reachobj$.
    $\tuple{s,J} \in \winregion\bomdp\reachobj$ implies that there is a strategy $\strat$ such that for all environments $j' \in J$, $\strat$ wins in the MDP ${\bomdp}_{j'}$, which includes $j$, which is a contradiction.
\end{proof}

\thmreachabilitycorrect*
\begin{proof}[Proof of~\cref{thm:reachability_correct}]
We show that $\tuple{s, J} \in \algreach(\bomdp, T) \Rightarrow \tuple{s, J} \in \winregion\bomdp\reachobj$ and that $\tuple{s, J} \not\in \algreach(\bomdp, T) \Rightarrow \tuple{s, J} \not\in \winregion\bomdp\reachobj$.

First we consider ($\tuple{s, J} \in \algreach(\bomdp, T) \Rightarrow \tuple{s, J} \in \winregion\bomdp\reachobj$).
All states $\tuple{s, J}$ returned by $\algreach(\bomdp, T)$ on line~\ref{line:reachreturn} have never been removed. %
Therefore, they have the property that $\forall j \in J\colon \tuple{s, J} \in \winregion{{\bomdp}_j}\reachobj$.
This is only possible if $A(\tuple{s, J})$ is non-empty.
During the algorithm, we modify the input BOMDP by removing states.
As the modified BOMDP is defined using fixpoint iteration, this availability of actions also holds for any belief $\tuple{s', J'}$ reached after playing an action in $A(\tuple{s, J})$. %
The loop terminates when there are no more losing states to remove.
Using this fact, we can construct a memoryless strategy $\strat$ which uniformly randomizes over the remaining actions in the modified BOMDP. 
In each belief we reach, we know that there exists a winning strategy for each environment and, thus, a path to reach the target.
This means that the probability of taking such a path is bounded by some positive probability $p_{\min} > 0$.
As we remain in $S$ for an infinite amount of steps, the probability of never taking such a path is $\lim_{n \to \infty} (1-p_{\min})^n = 0$, so we almost-surely reach the target.

Now consider ($\tuple{s, J} \not\in \algreach(\bomdp, T) \Rightarrow \tuple{s, J} \not\in \winregion\bomdp\reachobj$).
By~\cref{lem:removing_losing_states}, we know that removing losing states does not affect the winning region. 
We want to use this invariant by showing that we remove only losing states and show that this holds inductively. 
Initially, the modified BOMDP coincides with the original BOMDP.
We add only losing states via~\cref{lem:losing_inclusion} to $L$, and thus, we remove only losing states.
In later iterations, the losing states of the modified and the original BOMDP still coincide by our induction hypothesis, and thus, in that iteration, we again only remove losing states. 
Then, if $\tuple{s, J} \not\in \algreach(\bomdp, T)$, it was removed from the BOMDP by line 7 during some iteration and it is a losing state, so it is not in the winning region $\winregion\bomdp\reachobj$.
\end{proof}

\lemremovinglosingstates*
\begin{proof}[Proof of \cref{lem:removing_losing_states}]
By definition of a losing state, any strategy played will be losing.
Any strategy that has a positive probability of reaching a losing state is, therefore, losing.
The probability of a winning strategy visiting a losing state is, therefore, necessarily $0$.
Thus, removing losing states does not affect winning strategies, and therefore the winning region.
\end{proof}

\lemreachabilitycomplexity*
\begin{proof}[Proof of \cref{lem:reachability_complexity}]
    The outer loop on line~\ref{alg:reach:outerloop} removes at least 1 state per iteration, so the maximum number of iterations is $|S \times \powerset(I)|$, \ie, the number of states of the BOMDP.
    In each iteration of the inner loop on line~\ref{alg:reach:innerloop}, we compute almost-sure reachability on MDPs, $\winregion{{\bomdp}_i}{\reachobj}$, which is computable in polynomial time, see \eg~\cite[Corollary 10.107]{Baier2008}.
\end{proof}

\subsubsection*{5.3 Safe \Bu in BOMDPs}

\thmrabinpairpoly*
\begin{proof}[Proof of \cref{thm:rabin_pair_poly}]
\cref{lem:safe_buchi_to_buchi} shows that $\winregion\bomdp{\csafeobj \rabincobuchi \wedge \cbuchiobj \rabinbuchi}$ is computable as \Bu objective on $\custombomdp{(\staterestrict\memdp\rabincobuchi)}$.
\Bu objectives are computable in polynomial time, see \cref{lem:buchi_poly}.
\end{proof}

\lemsafebuechitobuechipreservationmemdp*
\begin{proof}[Proof of \cref{lem:safe_buechi_to_buechi_preservation_memdp}]
We show that winning in $\memdp$ implies winning in $\staterestrict\memdp\rabincobuchi$ and vice versa.
\begin{enumerate}
    \item Winning in $\memdp$ implies winning in $\staterestrict\memdp\rabincobuchi$.
Given that for a state $s \in S$ and initial belief $J$ we have that $\tuple{s, J} \models^{\memdp} \csafeobj \rabincobuchi \wedge \cbuchiobj \rabinbuchi$, we show that indeed $\tuple{s, J} \models^{\staterestrict\memdp\rabincobuchi} \csafeobj \rabincobuchi \wedge \cbuchiobj \rabinbuchi$.
As $\tuple{s, J}$ is winning in $\memdp$, there exists a winning strategy $\strat$.
We look at the induced MEMC $\memdp[\strat]$ and $\staterestrict\memdp\rabincobuchi[\strat]$ and show that they are equal.
For the initial states of $\memdp$ and $\staterestrict\memdp\rabincobuchi$, the actions are distributed according $\strat(\tuple{s, J})$. %
As $\strat$ is winning, we necessarily have that $s \in \rabincobuchi$.
Additionally, for all successors, we have that $s' \in \rabincobuchi$.
As the transition function of $\staterestrict\memdp\rabincobuchi$ is equal to that of $\memdp$ for transitions that remain in $\rabincobuchi$, so the outgoing transitions of the initial state are equal in both $\memdp$ and $\staterestrict\memdp\rabincobuchi$. 
We can inductively apply this logic to show that all transitions from states reachable from $\tuple{s, J}$ are equal, and thus that $\memdp[\strat] = \staterestrict\memdp\rabincobuchi[\strat]$.
    
    \item Winning in $\staterestrict\memdp\rabincobuchi$ implies winning in $\memdp$.
    This direction is trivial, as any winning strategy never visits the sink state, and therefore, the winning strategy can be mimicked in $\memdp$.
\end{enumerate}
\end{proof}

\lemsafebuechitobuechipreservationmemdptwo*
\begin{proof}[Proof of \Cref{lem:safe_buechi_to_buechi_preservation_memdp_two}]
As $\rabinbuchi$ cannot contain $\bot$, and $\bot$ is a deadlock state, it is clear that when a strategy $\strat$ satisfies $\cbuchiobj \rabinbuchi$, it also satisfies $\csafeobj \rabincobuchi$. The other direction is trivial.
\end{proof}

\lemsafebuchitobuchi*
\begin{proofsketch}[Proof of \cref{lem:safe_buchi_to_buchi}]
We prove the lemma in two steps:
\begin{enumerate}
    \item $\winregion\bomdp{\csafeobj \rabincobuchi \wedge \cbuchiobj \rabinbuchi} = \winregion{\custombomdp{(\staterestrict{\memdp}{\rabincobuchi})}}{\csafeobj \rabincobuchi \wedge \cbuchiobj \rabinbuchi}$.
    
    This statement follows from two applications of \cref{thm:obj_correspondence} and \cref{lem:safe_buechi_to_buechi_preservation_memdp}.
    \item $\winregion{\custombomdp{(\staterestrict{\memdp}{\rabincobuchi})}}{\csafeobj \rabincobuchi \wedge \cbuchiobj \rabinbuchi} = \winregion{\custombomdp{(\staterestrict{\memdp}{\rabincobuchi})}}{\cbuchiobj \rabinbuchi}$.

    This statement follows from two applications of \cref{thm:obj_correspondence} and \cref{lem:safe_buechi_to_buechi_preservation_memdp_two}. \qedhere
\end{enumerate}

\lembuchipoly*
\begin{proof}[Proof of \Cref{lem:buchi_poly}]
We can reduce any MEMDP $\memdp$ (and thus also any BOMDP) with an almost-sure \Bu objective $T$ to a MEMDP $\memdp'$ with an almost-sure reachability objective $T'$.
The reduction is similar to the one presented in~\cite{DBLP:journals/jacm/BaierGB12}.
First, we introduce a new state $\top$, which is the target of the reachability objective $T'=\{\top\}$.
The transition function of $\memdp'$ is equal to that of $\memdp$, except for transitions that reach states in $T$.
These transitions have their probability multiplied by $\frac12$, with the remaining probability being redirected to $\top$.
In $\memdp'$, a strategy reaches $\top$ almost-surely iff it visits states in $T$ infinitely often.
\end{proof}

\end{proofsketch}

\section{Proofs and Auxiliary Material of \cref{sec:pspace:algorithm}}\label{sec:appendix:local}

\subsubsection*{6.2 Local View on BOMDPs}
The following lemma makes preserving the graph preservation precise.
\begin{restatable}[]{lemma}{lemlocalgraphisomorphism}   
\label{lem:local_graph_isomorphism}
    For an environment $j \in J$, let $G_j \defeq \tuple{V_j, E_j}$ denote a labeled graph.
    Let $V_j \defeq S$, and $E_j \defeq \{ s \xrightarrow{a} s' \mid p_j(s, a, s') > 0 \}$.
    For all $j, j' \in J$, we have that $G_j = G_{j'}$.
\end{restatable}

\subsubsection*{6.3 An Algorithm for Localized Rabin Objectives}\label{appx:algorithm:loc:rabin}

\begin{restatable}[]{lemma}{lemimmediatelywinningiswinning}
\label{lem:immediately_winning_is_winning}
A state that is immediately winning is winning.
\end{restatable}

\begin{proof}[Proof of \cref{lem:immediately_winning_is_winning}]
Winning the objective $\csafeobj\rabincobuchi \wedge \cbuchiobj\rabinbuchi$ is a stronger condition than $\ccobuchiobj\rabincobuchi \wedge \cbuchiobj\rabinbuchi$.
Additionally, winning a fixed Rabin pair $\rabinObj_k$ is a stronger condition than winning any Rabin pair, as is the definition of winning.
\end{proof}

\begin{restatable}[]{lemma}{lemreachablefrombscc}\label{lem:reachable_from_bscc}
    Let $\mc = \localbomdp\memdp J[\strat]_j$ be a MC induced by strategy $\strat$ in environment $j \in J$.
    For all BSCCs $\Sbsccj j \subseteq S$ in $\mc$, all states inside the BSCC $s \in \Sbsccj j$ have the following property:
    For all environments $j' \in J$, the states reachable from $s$ in $\localbomdp\memdp J[\strat]_{j'}$ are in $\Sbsccj j \sqcup F$.
\end{restatable}
\noindent

\begin{proof}[Proof of \cref{lem:reachable_from_bscc}]
For contradiction, assume that we can reach a state $s \in S \setminus (\Sbsccj j \sqcup F)$ in environment $j' \in J$.
Clearly, this is not possible from a state in $F$, as these are sink states.
Therefore, we took a transition from a state $s' \in \Sbscc$ to a state $s'' \in S \setminus (\Sbsccj j \sqcup F)$ with some action $a$.
This implies that $p_{j'}(s', a, s'') > 0$.
However, as we know that $s'' \not\in \Sbsccj j \sqcup F$, we know that $p_j(s', a, s'') = 0$, and thus that $\beliefupd(\tuple{s', J}, a, s'') \ne \tuple{s'', J}$.
By definition of $\localbomdp\memdp J$, this implies that $p_j(s', a, \tuple{s', a, s''}) > 0$, contradicting the assumption that $\Sbsccj j \subseteq S$.
\end{proof}

\begin{restatable}[Reachable almost-sure winning states]{lemma}{lemwinningreachablewinning}\label{lem:winning_reachable_winning}
    All states reachable from a winning state $s$ under a winning strategy are winning.
\end{restatable}
\noindent
\begin{proofsketch}[Proof of \cref{lem:winning_reachable_winning}]
First, by definition, the winning strategy $\strat$ is winning for $s$.
Assume for contradiction that we can reach a state $s'$ which is not winning.
Because $s'$ is reachable, we have a path from $s$ to $s'$ with some probability $p$.
As $s'$ is not winning, we know that the probability of satisfying the objective from $s'$ is less than 1, say $v$.
Therefore, the probability of winning from $s$ is at most $1 - pv < 1$, which contradicts our assumption that $\strat$ is winning.
\end{proofsketch}

\begin{restatable}[]{lemma}{lemnosubbscc}\label{lem:no_sub_bscc}
    Let $\mc = \localbomdp\memdp J[\strat]_j$ be a MC induced by strategy $\strat$ in environment $j \in J$.
    For all BSCCs $\Sbsccj j \subseteq S$ in $\mc$, all states $s$ inside the BSCC $\Sbsccj j$ have the following property:
    For all environments $j' \in J$, the BSCC reached from $s$ cannot be a strict subset of $\Sbsccj j$.
\end{restatable}\noindent
\begin{proof}[Proof of \cref{lem:no_sub_bscc}]
Assume for contradiction, that $\strat$ reaches a BSCC $\Sbsccj{j'} \subset \Sbsccj j$ in environment $j'$.
Then, as $\Sbsccj j \subseteq S$, also: $\Sbsccj{j'} \subseteq S$.
According to \cref{lem:local_graph_isomorphism}, both $j$ and $j'$ must have the same underlying graph for transitions that remain inside of $S$.
As a consequence, the transitions (and therefore also states) belonging to $\Sbsccj j$ must also be present in $\Sbsccj{j'}$, contradicting the assumption that $\Sbsccj{j'} \subset \Sbsccj j$.
\end{proof}

\begin{restatable}[]{lemma}{lembsccdetermineswinning}
\label{lem:bscc_determines_winning}
    Given any MC $\mc$ with some BSCC $\Sbscc$ and any Rabin pair $\rabinPair i$. Either $\forall s \in \Sbscc\colon s \models^{\mc} \csafeobj\rabincobuchi_i \wedge \cbuchiobj\rabinbuchi_i$ or $\forall s \in \Sbscc \colon s \not\models^{\mc} \csafeobj\rabincobuchi_i \wedge \cbuchiobj\rabinbuchi_i$.
\end{restatable}\noindent
\begin{proof}[Proof of \cref{lem:bscc_determines_winning}]
    All states are reached infinitely often in a BSCC, so all paths $\Path$ that start in a state in $\Sbscc$, we have that $\mathit{Inf}(\Path) = \Sbscc$.
    Therefore, for all states $s \in \Sbscc$ we have that $s \models^{\mc} \rabinPair i \Leftrightarrow \Sbscc \subseteq \rabincobuchi_i \wedge \Sbscc \cap \rabinbuchi_i \ne \emptyset$.
\end{proof}

\lemcommonpairreachable*
\begin{proof}[Proof of \cref{lem:common_pair_reachable}]
A strategy that almost-surely reaches states which are (immediately, cf.\ \cref{lem:immediately_winning_is_winning}) winning is winning. 
It remains to show that from any winning state, all winning strategies almost-surely reach immediately winning states.

Given any strategy $\strat$, in every environment, we must almost-surely reach the BSCCs in the induced MC~\cite[Theorem 10.27]{Baier2008}. By \cref{lem:common_rabin_pair}, the BSCCs that are reached are immediately winning. 

\thmrabinalgcorrect*
\begin{proof}[Proof \Cref{thm:rabin_alg_correct}]
\cref{lem:common_pair_reachable,lem:winset_correct} together prove~\cref{thm:rabin_alg_correct}:
\begin{enumerate}
 \item $\winset$ computed by \rabinalgname{} is equal to the set of immediately winning states,  $\rabinWinSet$ (\cref{lem:winset_correct})
    \item A state is winning iff it can almost-surely reach the set $\rabinWinSet$(\cref{lem:common_pair_reachable}).
   \item Line~\ref{line:return} returns the states that can almost-surely reach $\winset$. \qedhere
\end{enumerate}%
\end{proof}

\lemcommonrabinpair*

\begin{proof}[Proof of \Cref{lem:common_rabin_pair}]
If the BSCC $\Sbsccj j$ reached is a frontier state, then by construction, we win with \emph{every} Rabin pair, so the statement trivially holds.
Otherwise, by construction of $\localbomdp{\memdp}{J}$, $\Sbsccj j$ must entirely be contained within $S$, \ie, $\Sbsccj j \cap F = \emptyset$.
Furthermore, as the strategy $\strat$ is winning, the BSCC $\Sbsccj j$ must win some Rabin pair $\rabinObj_i$ in environment $j$ (\cref{lem:bscc_determines_winning}).
It remains to show that the states in $\Sbsccj j$ are immediately winning $\rabinObj_i$ in \emph{every} environment. 
For $s \in \Sbsccj j$ in environment $j' \ne j$, we can distinguish several cases based on the fragment reachable from $s$:
\begin{enumerate}[I]
\item The reachable fragment of the environment $j'$ is exactly the same as $\Sbsccj j$. 
In this case, the same Rabin pairs are immediately winning~(\cref{lem:bscc_determines_winning}).
\item The reachable fragment of the environment $j'$ is a strict subset of $\Sbsccj j$.
This contradicts \cref{lem:no_sub_bscc}.
\item The reachable fragment of the environment $j'$ contains at least a state not in $\Sbsccj j$.
\end{enumerate}
Thus, we only need to consider this last case. 

We first note the following.
By \cref{lem:reachable_from_bscc}, all reachable states are in $\Sbsccj j$ or in $F$.
By \cref{lem:winning_reachable_winning}, all reachable states under a winning strategy are winning.
Thus, from state $s$, in every environment, the reachable states are a subset of the union of $\Sbsccj j$ and the set of winning frontier states. In particular, (1)~$\forall s \in \Sbsccj j\colon s \models^{\localbomdp{\memdp}{J}_{j'}} \csafeobj\rabincobuchi_i$.
First, there cannot be any BSCC that contains both states in $\Sbsccj j$ and $F$.
Second, by \Cref{lem:no_sub_bscc}, there is no BSCC strictly contained in $\Sbsccj j$.
Thus: In the union of $\Sbsccj j$ and $F$, the only BSCCs are frontier states.
In any MC, we must almost-surely reach a BSCC~\cite[Theorem 10.27]{Baier2008}. 
Combining this with the statements above, from $s$ onwards, we almost-surely reach winning frontier states. In these states, all Rabin pairs are immediately winning.
From this, we deduce: (2)~$
\forall s \in \Sbsccj j\colon s \models^{\localbomdp{\memdp}{J}_{j'}} \Finally(\csafeobj\rabincobuchi_i \wedge \cbuchiobj\rabinbuchi_i)$.

Thus (1) $\forall s \in \Sbsccj j\colon s \models^{\localbomdp{\memdp}{J}_{j'}} \csafeobj\rabincobuchi_i$ and (2)~$
\forall s \in \Sbsccj j\colon s \models^{\localbomdp{\memdp}{J}_{j'}} \Finally(\csafeobj\rabincobuchi_i \wedge \cbuchiobj\rabinbuchi_i)$.
From (2), we derive $
\forall s \in \Sbsccj j\colon s \models^{\localbomdp{\memdp}{J}_{j'}} \Finally( \cbuchiobj\rabinbuchi_i)$ which in turn implies\footnote{as eventually infinitely often implies infinitely often.} $
\forall s \in \Sbsccj j\colon s \models^{\localbomdp{\memdp}{J}_{j'}} \cbuchiobj\rabinbuchi_i$, which combines with (2) to show that each such $s$ is immediately winning. 
\end{proof}

\end{proof}

\lemwinsetcorrect*
\begin{proofsketch}[Proof of \cref{lem:winset_correct}]
    The algorithm iterates over each Rabin pair and modifies it to be a local Rabin objective.
    Next, it computes the states for which this pair is an immediate common winning Rabin pair.
    By taking the union of all these states (for each pair), we obtain $\winset = \rabinWinSet$.
\end{proofsketch}

\lemlocalrabinpoly*
\begin{proof}[Proof of \cref{lem:local_rabin_poly}]
The loop on line~\ref{line:forloop} has a fixed number of iterations, namely $n$, which is linear in the size of the input (as we encode MEMDPs as a list of MDPs).
In each iteration, we compute the union of sets, which are bounded by the size of the input.
Next, we compute the winning region of a safe \Bu objective, which, according to \cref{thm:rabin_pair_poly} is computable in polynomial time on BOMDPs.
We remark that constructing a BOMDP of a $J$-local MEMDP produces an identical MEMDP (up to renaming of states).
Finally, after the loop on line~\ref{line:forloop}, we compute the winning region of an almost-sure reachability objective, which, according to \cref{lem:reachability_complexity}, is computable in polynomial time.
\end{proof}
\subsubsection*{6.4 Recursive Computation of Winning Regions}

\thmlocalwinningmeansglobal*
\begin{proof}[Proof of \Cref{thm:local_winning_means_global}]
We prove this by comparing the probability measure of winning paths under a winning strategy $\strat$ in the $J$-local MEMDP $\localbomdp\memdp J$ and the BOMDP $\bomdp$.
We prove the two inclusions.

    \noindent For a given state $\tuple{s, J} \in (\winregion{\localbomdp{\memdp}{J}}{\clocRabinObj J} \cap S) \times \{J\}$, we show that  $\tuple{s, J} \in \winregion{\memdp}{\rabinObj}_{J}$.
    
    As $\tuple{s, J}$ is winning, there exists a winning strategy $\strat$ that is winning $\clocRabinObj J$ on $\localbomdp\memdp J$.
    Therefore:
    \begin{align*}
    \int_{\Path \in \clocRabinObj J} p^\strat_i(\Path)\,d\Path = 1 &=
    \int_{\substack{\Path \in \clocRabinObj J \\ \Path \in \creachobj F}} p^\strat_i(\Path)\,d\Path + \int_{\substack{\Path \in \clocRabinObj J \\ \Path \not\in\creachobj F}} p^\strat_i(\Path)\,d\Path\\
    &= \int_{\substack{\Path \in \clocRabinObj J \\ \Path \in \creachobj F}} p^\strat_i(\Path)\,d\Path + \int_{\substack{\Path \in \clocRabinObj J \\ \Path \in\csafeobj S}} p^\strat_i(\Path)\,d\Path.
    \end{align*}
        Let $\mathit{CWF} \defeq \mathit{WinLocal}_J(\mathit{RF}, \winregion{\memdp}{\rabinObj}_{\subset J})$.
    As $\strat$ is winning, we obtain:
    \begin{align}
    \label{eq:targetorso}
    1= \int_{\substack{\Path \in \clocRabinObj J \\ \Path \in \creachobj F}} p^\strat_i(\Path)\,d\Path + \int_{\substack{\Path \in \clocRabinObj J \\ \Path \in\csafeobj S}} p^\strat_i(\Path)\,d\Path =
    \int_{\substack{\Path \in \clocRabinObj J \\ \Path \in \creachobj \mathit{CWF}}} p^\strat_i(\Path)\,d\Path + \int_{\substack{\Path \in \clocRabinObj J \\ \Path \in\csafeobj S}} p^\strat_i(\Path)\,d\Path.
    \end{align}
For the states in $\mathit{CWF}$ there is a strategy $\strat_f$ that wins globally in $\bomdp$.

    Now consider a strategy $\strat_g$ for the global BOMDP, which mimics $\strat$ while the belief is $J$, but switches to $\strat_f$ once a frontier state is reached.
    What remains to be shown is that $\strat_g$ indeed has probability 1 of satisfying $\rabinObj$.
    We denote the transition function of $\bomdp$ as $p'$ to distinguish it from $p$. Using a similar idea as above, we want to show:
    \begin{align}
    \label{eq:something1}
   1 \stackrel?= \int_{\Path \in \rabinObj} {p'}^{\strat_g}_i(\Path)\,d\Path = 
    \int_{\substack{\Path \in \rabinObj \\ \pathbelief(\Path) = J}} {p'}^{\strat_g}_i(\Path)\,d\Path + 
    \int_{\substack{\Path \in \rabinObj \\ \pathbelief(\Path) \ne J}} {p'}^{\strat_g}_i(\Path)\,d\Path.
    \end{align}
    We remark that paths that satisfy $\pathbelief(\Path) = J$ can be mapped one-to-one to paths in $\localbomdp\memdp J$ that satisfy $\csafeobj S$.
    Additionally, this means that $\clocRabinObj J$ is satisfied. Next, we observe that $p=p'$ for paths that have belief $J$. Thus 
     \begin{align*}
    \int_{\substack{\Path \in \rabinObj \\ \pathbelief(\Path) = J}} {p'}^{\strat_g}_i(\Path)\,d\Path =  \int_{\substack{\Path \in \clocRabinObj J \\ \Path \in\csafeobj S}} p^\strat_i(\Path)\,d\Path. 
    \end{align*}
    By substituting this in \eqref{eq:something1} we obtain:
    \begin{align*}
     1 \stackrel?= 
    \int_{\substack{\Path \in \clocRabinObj J \\ \Path \in\csafeobj S}} p^\strat_i(\Path)\,d\Path + 
    \int_{\substack{\Path \in \rabinObj \\ \pathbelief(\Path) \ne J}} {p'}^{\strat_g}_i(\Path)\,d\Path.
    \end{align*}
Using \eqref{eq:targetorso}, we can reformulate this to:
\begin{align}\label{eq:eq3}
    \int_{\substack{\Path \in \rabinObj \\ \pathbelief(\Path) \ne J}} {p'}^{\strat_g}_i(\Path)\,d\Path \stackrel?=
    \int_{\substack{\Path \in \clocRabinObj J \\ \Path \in \creachobj \mathit{CWF}}} p^\strat_i(\Path)\,d\Path.
\end{align}
For (absorbing) winning frontiers, we can rewrite the right-hand side:
\[
\int_{\substack{\Path \in \clocRabinObj J \\ \Path \in \creachobj \mathit{CWF}}} p^\strat_i(\Path)\,d\Path =
\sum_{f \in \mathit{CWF}}\int_{\substack{\Path \in \clocRabinObj J \\ \Path \in \creachobj \{f\}}} p^\strat_i(\Path)\,d\Path
\]
Note that all paths that reach $\mathit{CWF}$ do satisfy $\clocRabinObj J$. 
\begin{align}
\label{eq:rewrittenrhs}
\int_{\substack{\Path \in \clocRabinObj J \\ \Path \in \creachobj \mathit{CWF}}} p^\strat_i(\Path)\,d\Path =
\sum_{f \in \mathit{CWF}}\int_{\substack{ \Path \in \creachobj \{f\}}} p^\strat_i(\Path)\,d\Path,
\end{align}
which is the definition of the reachability probability for states in $F$ (given that they are absorbing). 

We similarly aim to rewrite the left side. Note that any path where the belief eventually changes must hit a (first) revealing transition, \ie, let $\Path = \Path' \cdot as'$ such that $\pathbelief(\Path') = J$ and $\pathbelief(\Path) = J' \neq J$. We can collect the states for which such paths exist in a set $F'$.
    Thus, we can write the probability of satisfying the objective $\rabinObj$ as the sum of reaching such a frontier state multiplied by the probability of satisfying the objective \emph{after} reaching the frontier state:
    \begin{align*}
    \int_{\substack{\Path \in \rabinObj \\ \pathbelief(\Path) \ne J}} {p'}^{\strat_g}_i(\Path)\,d\Path =
    \sum_{f \in F'} &\int_{\Path \in \creachobj \{ f\}} {p'}^{\strat}_i(\Path)\,d\Path \cdot \int_{\Path \in \rabinObj} {p'}^{\strat_f}_i(\Path \mid \first(\pi) = f) \,d\Path.
    \end{align*}
  We note that this notation requires a slightly generalized notion of the measure $ {p'}^{\strat_f}_i$ where we can choose the initial state. Now, as $\strat_f$ is winning in the BOMDP from the frontier states, we know that $\int_{\Path \in \rabinObj} {p'}^{\strat f}_i(\Path \mid \first(\pi) = f) \,d\Path = 1$, thus 
  \begin{align*}
    \int_{\substack{\Path \in \rabinObj \\ \pathbelief(\Path) \ne J}} {p'}^{\strat_g}_i(\Path)\,d\Path =
    \sum_{f \in F'} &\int_{\Path \in \creachobj \{ f\}} {p'}^{\strat}_i(\Path)\,d\Path \cdot 1.
    \end{align*}
    The paths do not update their belief until they reach the target, so we have that $p=p'$ in the above equality.
  \begin{align}\label{eq:rewrittenlhs}
    \int_{\substack{\Path \in \rabinObj \\ \pathbelief(\Path) \ne J}} {p'}^{\strat_g}_i(\Path)\,d\Path =
    \sum_{f \in F'} &\int_{\Path \in \creachobj \{ f\}} {p}^{\strat}_i(\Path)\,d\Path.
    \end{align}
    Because of the definition of $F'$, it is equal to the reachable frontiers in $\localbomdp\memdp J$.
    Combining \eqref{eq:rewrittenrhs} and \eqref{eq:rewrittenlhs} in \eqref{eq:eq3} yields the following equality, as desired:
    \sj{thanks! needs rewording here:-)}
\begin{align*}
    \int_{\substack{\Path \in \rabinObj \\ \pathbelief(\Path) \ne J}} {p'}^{\strat_g}_i(\Path)\,d\Path &=
\sum_{f \in \mathit{CWF}}\int_{\substack{\Path \in \clocRabinObj J \\ \Path \in \creachobj \{f\}}} p^\strat_i(\Path)\,d\Path\\
&=  \sum_{f \in F'} \int_{\Path \in \creachobj \{ f\}} {p}^{\strat}_i(\Path)\,d\Path\\
&= \int_{\substack{\Path \in \clocRabinObj J \\ \Path \in \creachobj \mathit{CWF}}} p^\strat_i(\Path)\,d\Path.
\end{align*}
    For the other direction of the proof, \ie, for a give state $\tuple{s, J} \in \winregion{\memdp}{\rabinObj}_{J}$, show that indeed  $(\tuple{s, J} \in \winregion{\localbomdp{\memdp}{J}}{\clocRabinObj J} \cap S) \times \{J\}$, we use a similar argument.
    A winning strategy $\strat_g$ on the BOMDP can be decomposed into a local strategy $\strat$ and a frontier strategy.
    The above equalities follow.

\end{proof}

\thmcheckcorrect*
\begin{proof}[Proof of \Cref{thm:check_correct}]

By \cref{cor:memdp_win_region_characterisation}, we know that the winning region can be computed by looking at $J$-local MEMDPs.
For each of these $J$-local MEMDPs, we compute the local winning region, which, according to \cref{thm:local_winning_means_global}, is equal to the global winning region when we use the correct localized Rabin objective $\clocRabinObj J$.
This requires the winning frontier to be equal to $\mathit{WinLocal}_J(\mathit{RF}, \winregion{\memdp}{\rabinObj}_{\subset J})$, \ie, all frontier states that are globally winning.
These are recursively computed on line~\ref{line:wf}.
Eventually, we reach a belief for which there are no more reachable frontiers, in which case $\mathit{RF}$ will be empty.
In this case, the computed $\mathit{WF}$ will be empty, too, causing the recursion to end.
Thus, the winning region is correct, and consequently, the winning frontier is to be used higher up in the recursion.

\end{proof}

\thmcheckspace*
\begin{proof}[Proof of \cref{thm:check_space}]
We give an upper bound on the space complexity of \genericalgoname{}.
As \genericalgoname{} is a recursive algorithm, we can compute an upper bound by multiplying the space complexity of a single recursive step by the maximum recursion depth.
The size of the local MEMDP and the winning frontier are both polynomial in the input.
The maximum recursion depth is equal to $I$, as the belief-support is strictly decreasing in size with each recursive step.
\end{proof}

\thmalgbigo*
\begin{proof}[Proof of \Cref{thm:alg_big_o}]
The algorithm performs the recursive step for all reachable frontier states, of which there are at most $|S|^2\cdot|A|$.
This recursive step decreases the size of the belief by at least $1$, so the recursion is at most $|I|$ levels deep.
In total, this means that we have to perform the recursive step at most $(|S|^2\cdot|A|)^{|I|}$ times.

The time complexity of each recursive step is polynomial in $|\memdp|$ and $|\rabinObj|$ (see~\cref{lem:local_rabin_poly}).
\end{proof}

\cortwomemdps*
\begin{proof}[Proof of \cref{cor:2memdps}]
Follows immediately from~\cref{thm:alg_big_o} by substituting $|I|$ with $k$.
\end{proof}

\end{document}